\documentclass[aps,amssymb,amsmath,amsfonts,twocolumn,pra]{revtex4}

\usepackage{graphicx}
\usepackage{dcolumn}
\usepackage{bm}
\usepackage{bbm}
\usepackage{psfrag}
\usepackage{epsfig}
\usepackage{amsmath}
\usepackage{amssymb}
\usepackage{color}
\usepackage{mathbbol}
\usepackage{subfigure}
\usepackage{amsthm}
\usepackage{bbold}
\usepackage{latexsym}

\usepackage{hyperref}
\usepackage{graphicx}
\usepackage{graphics}
\usepackage{color}

\usepackage{multirow}

\newtheorem{theorem}{Theorem}

\renewcommand{\c}[1]{\mathcal{#1}}

\newcommand{\1}{\mathbbm{1}} 
\newcommand{\idty}{\1}

\DeclareMathOperator*{\tr}{Tr}

\renewcommand{\>}{\rangle}
\providecommand{\abs}[1]{|#1|}

\newcommand{\beq}{\begin{equation}}
\newcommand{\eeq}{\end{equation}}
\newcommand{\ket}[1]{\ensuremath{\left|{#1}\right\rangle}}
\newcommand{\bra}[1]{\ensuremath{\left\langle{#1}\right |}}
\renewcommand{\rho}{\varrho}

\begin{document}

\title{Device-independent quantum reading and \\
       noise-assisted quantum transmitters}
\author{W. Roga$^{1}$, D. Buono$^{1}$, and F. Illuminati$^{1,2,3}$\footnote{Corresponding author: illuminati@sa.infn.it}}
\affiliation{$^1$ Dipartimento di Ingegneria Industriale, Universit\`a degli Studi di
Salerno, Via Giovanni Paolo II 132, I-84084 Fisciano (SA), Italy}
\affiliation{$^2$ INFN, Sezione di Napoli, Gruppo collegato di Salerno, I-84084 Fisciano
(SA), Italy}

\begin{abstract}
In quantum reading, a quantum state of light (transmitter) is applied to
read classical information. In the presence of noise or for sufficiently
weak signals, quantum reading can outperform classical reading by enhanced
state distinguishability. Here we show that the enhanced quantum efficiency
depends on the presence in the transmitter of a particular type of quantum
correlations, the discord of response. Different encodings and transmitters
give rise to different levels of efficiency.
Considering noisy quantum probes we show that squeezed thermal
transmitters with non-symmetrically distributed noise among the field modes
yield a higher quantum efficiency compared to coherent
thermal quantum states.
The noise-enhanced quantum advantage is a
consequence of the discord of response being a non-decreasing function of
increasing thermal noise under constant squeezing, a behavior that leads to
an increased state distinguishability. We finally show that, for non-symmetric
squeezed thermal states, the probability of error, as measured by the quantum Chernoff bound,
vanishes asymptotically with increasing {\em local} thermal noise at finite {\em global}
squeezing. Therefore, at {\em fixed} finite squeezing, noisy but strongly discordant quantum states
with large noise imbalance between the field modes can outperform noisy classical resources
as well as pure entangled transmitters with the same finite level of squeezing.
\end{abstract}

\date{December 02, 2014}
\maketitle

\section{Introduction}

In the context of quantum information and quantum technology the idea of
reading classical data by means of quantum states arises quite naturally~\cite{Pirandola2011a,Pirandola2011b}. In general, the standard
implementations of reading are based on optical technologies: the task is
the readout of a digital optical memory, where information is stored by
means of the optical properties of the memory cells that are in turn probed
by shining light, e.g. a laser beam, on them. The probing light is usually
denoted as the \emph{transmitter}. Interesting features arise in the regime
in which the transmitter has to be treated quantum mechanically. The maximum
rate of reliable readout defines the quantum reading capacity~\cite{Pirandola2011b}. The latter
can overcome the classical reading capacity, obtained by probing with
classical light, in several relevant settings. The (possibly quantum)
transmitter that is needed to extract the encoded information is prepared in
some initial state. By scanning a particular cell the transmitter changes
its properties in a way depending on the cell. The task is to recognize
which cell occurs based on the output state of the transmitter after it has
been detected and measured. Therefore the problem of reading reduces to the
problem of distinguishing the output states of the transmitter.

In such optical settings one needs to consider two main coding protocols
depending on the trade off between energy and coherence of the transmitters
and the channels that are being used. The first protocol is the so-called
\emph{amplitude shift keying} (ASK) in which the changes in the state of the
transmitter are caused by the cell-dependent losses in the intensity of the
transmitted signal~\cite{Pirandola2011a,Pirandola2011b,Lupo2013,Spedalieri2012,Tej2013}. The second main
protocol is the so-called \emph{phase shift keying} (PSK)~\cite{Hirota2011,Dallarno2012,Nair2011, Guha2013, Bisio2011,Dallarno2012a}.
This is a type of coding which does not produce
energy dissipation. On the other hand, it requires a very high coherence of
the transmitter, a feature that might be realized in realistic implementations~\cite{Dallarno2012}.

If the transmitter is quantum, the cells play the role of effective quantum
channels. The ASK protocol then corresponds to a dissipative channel coding,
while the PSK protocol is a particular case of unitary coding corresponding to a unitary channel.
Within the ASK protocol, it can be shown that in the low-energy regime there is an energy
threshold above which the maximally entangled transmitter, i.e. a two-mode
squeezed state, yields a better reading efficiency than any of the classical states with the same
energy~\cite{Pirandola2011a,Pirandola2011b,Spedalieri2012}. The general result is still
valid in the presence of some noise-induced decoherence. Within the ASK
protocol, coding is then realized by local channels corresponding to cells
with different reflectivities.

In the PSK protocol, the coding is realized by means of local unitary
operations, specifically local phase shifts~\cite{Hirota2011,Dallarno2012}. In the ideal,
noise-free protocol the transmitter is taken to be in a pure Gaussian
quasi-Bell state, i.e. a Bell-type superposition of quasi-orthogonal
coherent states. In this scheme, the resulting quantum advantage is
absolute, in the sense that quantum reading of the classical information
encoded via a phase shift of $\pi$ is achieved with vanishing error, while
any classical state of the transmitter always yields a finite error
probability.

In both the ASK and PSK protocols the transmitter is assumed to be a
bipartite system such that only one part of it scans the memory cell. This
choice is motivated by the fact that it maximizes distinguishability at the
output when the state of the transmitter is quantum. As already mentioned,
the reading efficiency is characterized by the probability of error.
Information is encoded in binary memory cells with indices $0$ and $1$. It
is thus written using only two local channels that are assumed to occur with
equal \emph{a priori} probabilities. Given the bi-partite input transmitter
$\rho _{AB}$, the two possible output states will be denoted by
$\rho_{AB}^{(0)}$ and $\rho _{AB}^{(1)}$.

The probability of error in distinguishing the two output states when
reading a memory cell by means of the same input $\rho_{AB}$ is given by the
well-known Helstrom formula \cite{Helstrom1976}:
\begin{equation}
P_{err} = \frac{1}{2} - \frac{1}{4}d_{Tr}\left(\rho_{AB}^{(0)},\rho_{AB}^{(1)}\right) \, ,
\label{Helstrom}
\end{equation}
where $d_{Tr} \equiv \|\rho_{AB}^{(0)} - \rho_{AB}^{(1)}\|_{Tr}$ is the trace distance,
with $\|X\|_{Tr}=\tr\sqrt{XX^{\dagger}}$. With our normalization convention, the trace distance ranges from
$0$ to a maximum of $2$ for orthogonal pure states.

In the original reading protocols the goal is to minimize $P_{err}$ over the
set of possible transmitter states $\rho_{AB}$ at fixed encoding in the
memory cells~\cite{Pirandola2011b,Dallarno2012}. The problem is thus
dependent on the type of memory device being used.

Here instead we wish to
provide a device-independent characterization of a given transmitter by
considering the worst-case scenario that maximizes the probability of error
$P_{err}$ over all possible codings. Once the worst-case scenario is
identified, one can then compare different classes of transmitter states to
identify the ones that minimize the maximum probability of error $P_{err}^{(\max)}$.

We will show
that the maximum probability of error $P_{err}^{(\max)}$ is a monotonically decreasing function of
the amount of quantum correlations present in the transmitter state
$\rho_{AB}$, as quantified by a recently introduced measure of quantum correlations,
the so-called discord of response~\cite{Roga} in its Gaussian version~\cite{Buono2014}
(for general reviews on quantum correlations and distinguishability of quantum states and
on discord-like correlations see~\cite{Spehner2014,Modi2012}). As a
consequence, every state with non-vanishing discord of response is able to
read any type of memory device with maximal $P_{err}<1/2$. On the other
hand, for each classical transmitter, i.e. for transmitter with
vanishing discord of response, there will always exist at least one memory
device which is completely invisible, i.e. for which $P_{err}=1/2$.
In these considerations we exclude the situations when two channels are chosen to be
arbitrarily similar. In this case the probability of error always approaches $1/2$
independently on the chosen transmitter.

In Sec.~\ref{errorprobdiscord} we derive the exact analytical relation
between the maximum probability of error and the Gaussian discord of response. In Sec.~\ref{Gausstrans}
we discuss the properties of classical and quantum Gaussian transmitters, comparing squeezed-thermal, thermal-squeezed
and coherent-thermal states; in subsection \ref{GausstransA} we derive upper and lower bounds on
the maximum probability of error, and in subsection \ref{phaseshiftreading} we identify the unitary coding
that maximizes the quantum Chernoff bound, namely the upper bound on the maximum
probability of error, for the classes of quantum Gaussian transmitters considered.
The unitary channel which maximizes the bounds on the probability of error turns out to be a particular PSK coding,
namely the unique traceless one, which is realized for a $\pi/2$ phase shift.

In Sec.~\ref{comparisonclass} we compare the
performance of Gaussian quantum states of light with classical states (coherent thermal states).
We show that strongly discordant squeezed thermal states possess a higher
reading efficiency than the corresponding classical states of light, that is
non-discordant Gaussian coherent thermal transmitters with the same total number
of photons (fixed energy). This realizes an important instance of
noise-enhanced quantum advantage over the corresponding noisy classical resources.

In Sec.~\ref{comparisonquant} we compare different families of discordant
Gaussian states, the squeezed thermal and the thermal squeezed states at fixed total number of photons
or squeezing. While for both classes of states the entanglement obviously
decreases with increasing thermal noise, we show that for squeezed thermal
states the discord is an increasing function of the number of thermal
photons at fixed squeezing, while the opposite holds for thermal squeezed states. Moreover,
for squeezed thermal transmitters, the quantum Chernoff bound is independent
of thermal noise. As a consequence, this type of transmitter plays a privileged
role in the considered class of quantum Gaussian resources because the associated
quantum efficiency is either enhanced or unaffected by increasing the thermal
noise. Thus squeezed thermal transmitters realize an instance of noise-enhanced
or noise-independent quantum resources at fixed squeezing. Both in the classical-to-quantum
comparison and in the comparison of different quantum resources, the key enhancement of
quantum advantage is realized in the situation of strongest asymmetry of the distribution
of thermal noise among the field modes: {\em local} noise enhancement leads
to {\em global} enhancement of quantum correlations.

The main results are summarized and some outlook perspectives for future work are discussed
in Sec.~\ref{summary}. Detailed calculations and auxiliary reasonings are reported in four Appendixes.

\section{Probability of error, bounds, and discord of response}

\label{errorprobdiscord}

In protocols of quantum reading with unitary coding the
two local channels acting on the input probe state $\rho_{AB}$ are unitary,
and denoted as $U^{(0)}_A$ and $U^{(1)}_A$. Therefore, in this type of protocol
$\rho_{AB}^{(0)} = U_A^{(0)}\rho_{AB}U_A^{(0)\dagger}$ and
$\rho_{AB}^{(1)} = U_A^{(1)}\rho_{AB}U_A^{(1)\dagger}$, so that the
probability of error reads
\begin{equation}
P_{err} = \frac{1}{2} - \frac{1}{4}d_{Tr}\left(U_A^{(0)}\rho_{AB}U_A^{(0)\dagger},
U_A^{(1)}\rho_{AB}U_A^{(1)\dagger}\right) \, .
\end{equation}
Since the trace norm is invariant under local unitary transformations, one
has equivalently
\begin{equation}
P_{err} = \frac{1}{2} - \frac{1}{4}d_{Tr}\left(\rho_{AB}, \widetilde{\rho }_{AB} \right) \, ,
\label{proberr}
\end{equation}
where $\widetilde{\rho }_{AB} = W_{A} \rho _{AB} W_{A}^{\dagger }$ and $W_A = U_A^{(0)\dagger}U_A^{(1)}$ is still a local unitary transformation
acting on the transmitted subsystem $A$. The absolute upper bound for the
probability of error is thus $1/2$, corresponding to a situation in which there
is no way to distinguish the two output states and therefore the memory
device becomes completely invisible to the transmitter.

In general, computing the trace distance proves to be extremely challenging~\cite{Audenaert2007}, even more so for Gaussian states of infinite-dimensional continuous-variable systems~\cite{Pirandola2008}. Therefore one has to look for analytically computable {\em a priori} upper and lower bounds. A natural upper bound on the probability of error, Eq.~(\ref{Helstrom}), in distinguishing two states $\rho_1$ and $\rho_2$ occurring with the same probability is provided by the quantum Chernoff bound $QCB$~\cite{Calsamiglia2008}:
\begin{equation}
P_{err} \leq QCB \equiv \frac{1}{2} \left[ \inf_{t\in (0,1)} \tr (\rho_1^{t}\rho_2^{1-t}) \right] \; .
\label{QCB}
\end{equation}
If the states $\rho_1$ and $\rho_2$ are not arbitrary, but they are qubit-qudit states related by a local, single-qubit unitary transformation, as in the case of the quantum reading protocol with unitary coding, for which $\rho_1 = \rho_{AB}$ and $\rho_2 = \widetilde{\rho }_{AB}$, then the quantum Chernoff bound $QCB$ is achieved for $t=1/2$ in Eq.~(\ref{QCB}), as discussed in Appendix~\ref{chbm}:
\begin{equation}
QCB = \frac{1}{2}\left[ \tr \left( \sqrt{\rho _{AB}} \sqrt{\widetilde{\rho }_{AB}} \right) \right] \; ,
\label{QCBLU}
\end{equation}
and its expression coincides with the quantum Bhattacharyya coefficient \cite{Pirandola2008,Sihui2008}, which provides an upper bound on $QCB$ for
arbitrary quantum states. The same result, Eq.~(\ref{QCBLU}), holds for Gaussian states related by traceless local symplectic transformations, as discussed in the following sections and in Appendix~\ref{chbm}.

Next, considering the Uhlmann fidelity yields a complete hierarchy of lower and upper bounds~\cite{Pirandola2008}:
\begin{equation}
LBP_{err} \leq P_{err} \leq QCB \; ,
\label{errorineq}
\end{equation}
where the lower bound on the probability of error $LBP_{err} \equiv (1-\sqrt{1-{\cal{F}}})/2$, and
the Uhlmann fidelity ${\cal{F}}$ between two quantum states $\rho_1,\rho_2$ is defined as
$\c F(\rho_1,\rho_2) \equiv \left( \tr{\sqrt{\sqrt{\rho_1}\rho_2\sqrt{\rho_1}}} \right)^2$.

For the quantum reading protocol with unitary coding, let us consider the maximum probability of error
in distinguishing the output of a binary memory cell encoded using one identity
and one arbitrary unitary channel $W_{A}$ chosen in the set of local unitary
operations with non-degenerate harmonic spectrum. The latter is the spectrum
of the complex roots of the unity and its choice is motivated by observing
that it excludes unambiguously the identity from the set of possible
operations: indeed, unitary operations with harmonic spectrum are orthogonal
(in the Hilbert-Schmidt sense) to the identity. We further assume that the
coding is unbiased, that is the two channels are equiprobable.

The worst-case scenario is defined by the probability of error Eq.~(\ref{proberr}) being the largest possible:
\begin{equation}
P_{err}^{(\max )}\equiv \max_{\{W_{A}\}}P_{err}=\frac{1}{2}-\frac{1}{4} \min_{\{W_{A}\}}d_{Tr}\,\left( \rho _{AB},\widetilde{\rho }_{AB}\right) .
\label{PerrMax}
\end{equation}

Let us now consider a recently introduced measure of quantum correlations, the so-called
discord of response~\cite{Roga}:
\begin{equation}
{\mathcal{D}}_{R}^{x}(\rho _{AB})\equiv \min_{\{W_{A}\}}{\cal{N}}_x^{-1}d_{x}^{2}\left(\rho_{AB},\widetilde{\rho }_{AB}\right) \, ,
\label{DiscOfResp}
\end{equation}
where the index $x$ denotes the possible different types of well behaved, contractive metrics under completely positive and trace-preserving
(CPTP) maps. The normalization factor ${\cal{N}}_x$ depends on the given metrics and is chosen in such a way to assure that ${\mathcal{D}}_{R}^{x}$
varies in the interval $[0,1]$. Finally, the set of local unitary operations $\{W_A\}$ includes all and only those local unitaries with
harmonic spectrum.

In the following, we will need to consider both the probability of error and different types of upper and lower bounds on it. Therefore we will be concerned with three different discords of response corresponding to three types of contractive distances: trace, Hellinger, and Bures.

The trace distance $d_{Tr}$
between any two quantum states $\rho_1$ and $\rho_2$ is defined as:
\begin{equation}
d_{Tr}\left( \rho_1, \rho_2 \right) \equiv \tr \left[ \sqrt{ \left( \rho_1 - \rho_2 \right)^{2}} \right] \, .
\label{tracedistance}
\end{equation}
The Bures distance, directly related to the fidelity ${\cal{F}}$, is defined as:
\begin{equation}
d_{Bu}\left( \rho_1, \rho_2 \right) \equiv \sqrt{ 2 \left( 1 - \sqrt{ {\cal{F}}(\rho_1,\rho_2) } \right) } \; \, .
\label{Buresdistance}
\end{equation}
Finally, the Hellinger distance is defined as:
\begin{equation}
d_{Hell}\left( \rho_1, \rho_2 \right) \equiv \sqrt{ \tr \left[ \big( \sqrt{\rho_1} - \sqrt{\rho_2} \big)^{2} \right] } \; \, .
\label{helldistance}
\end{equation}
For each discord of response, trace, Hellinger, and Bures, the normalization factor in Eq.~(\ref{DiscOfResp}) is, respectively:
${\cal{N}}_{Tr}^{-1} = 1/4$, ${\cal{N}}_{Hell}^{-1} = {\cal{N}}_{Bu}^{-1} = 1/2$.

If the two states $\rho_1$ and $\rho_2$ are bipartite Gaussian states related by local traceless symplectic transformations or bipartite qubit-qudit states related by a local single-qubit unitary operation, that is $\rho_1 = \rho_{AB}$ and
$\rho_2 = \widetilde{\rho }_{AB} = W_{A} \rho _{AB} W_{A}^{\dagger }$, then, by exploiting Eq.~(\ref{QCBLU}),
it is straightforward to show that the quantum Chernoff bound is a simple, monotonically non-increasing simple
function of the Hellinger distance:
\begin{equation}
QCB = \frac{1}{4}\left( 2 - d_{Hell}^{2}(\rho_{AB},\widetilde{\rho }_{AB}) \right) \; .
\end{equation}
It is then immediate to show that the maximum of $QCB$ over the set of local unitary operations $\{W_{A}\}$ with completely
non-degenerate harmonic spectrum is a simple linear function of the Hellinger discord of response:
\begin{equation}
QCB^{max} = \frac{1}{2}\big( 1 - {\cal D}_{R}^{Hell}(\rho _{AB}) \big) \; .
\end{equation}
The discord of response quantifies the response of a quantum state to least-disturbing local unitary perturbations
and satisfies all the basic axioms that must be obeyed by a \emph{bona fide} measure of quantum correlations~\cite{Roga}: it vanishes if
and only if $\rho _{AB}$ is a classical-quantum state; it is invariant under local unitary operations; by fixing a well-behaved metrics
such as trace, Bures, or Hellinger, it is contractive under CPTP maps on subsystem $B$, i.e. the subsystem that is not perturbed by the
local unitary operation $W_A$; and reduces to an entanglement monotone for pure states, for one of which it also assumes the maximum possible value ($1$).

By comparing Eqs.~(\ref{PerrMax}) and (\ref{DiscOfResp}) with $x=Tr$, it is immediate to relate the maximum probability of error $P_{err}^{(\max )}$
to the trace discord of response ${\mathcal{D}}_{R}^{Tr}$:
\begin{equation}
P_{err}^{(\max )}=\frac{1}{2}-\frac{1}{2}\sqrt{{\mathcal{D}}_{R}^{Tr}(\rho_{AB})} \; \, .
\label{perrdiscord}
\end{equation}
From Eq.~(\ref{perrdiscord}) it follows that half of the square root of the
trace discord of response yields the difference between the absolute maximum of
the probability of error (i.e. $1/2$) and the maximum probability of error at
fixed transmitter state $\rho_{AB}$.

A vanishing trace discord of response implies that there exists at least one memory that
cannot be read by classical-quantum transmitters. A maximum trace discord of response
(${\mathcal{D}}_{R}^{Tr}=1$) implies that, irrespective of the coding, the maximally
entangled transmitter will read any memory without errors: indeed, any local
unitary operation with harmonic spectrum transforms a maximally entangled
state into another maximally entangled state orthogonal to it, and therefore yields perfect
distinguishability at the output.


\section{Quantum reading with squeezed thermal states} \label{Gausstrans}

In the following, in order to compare the efficiency of classical (non-discordant) and quantum noisy sources of light in reading protocols, we will consider
two-mode Gaussian states of the electromagnetic field. The states with vanishing first moments of the quadratures are fully described by their covariance matrix $\sigma$~\cite{Weedbrook2012,Salerno,Ferraro}:
\begin{equation}
\sigma = \frac{1}{2}\begin{bmatrix}
a & 0 & c_1 & 0 \\
0 & a & 0 & c_2 \\
c_1 & 0 & b & 0 \\
0 & c_2 & 0 & b
\end{bmatrix} \; ,
 \label{sqthcorm}
\end{equation}
The range of values of $a, b, c_1$ and $c_2$ for which the corresponding states are physical (i.e. correspond to positive density matrices)
is determined by the Heisenberg uncertainty relation stated in symplectic form:
\begin{equation}
\sigma+\frac{i}{2}\omega\oplus\omega\geq 0 \; ,
\label{uncert}
\end{equation}
where $\omega=\begin{bmatrix}
0 & 1 \\
-1 & 0
\end{bmatrix}$ is the symplectic form. In all the paper, if we refer to the symmetric states we mean $a=b$.
In the following we will focus on two rather general classes of (undisplaced) Gaussian states, the squeezed thermal states (STS)
and the thermal squeezed states (TSS). The former are defined by two-mode squeezing
$S(r)=\exp {\{ra_{1}^{\dagger }a_{2}^{\dagger }-r^{\ast }a_{1}a_{2}\}}$ applied on, possibly non-symmetric, two-mode thermal
states. Notice that the denomination STS is sometimes used in the literature to denote any Gaussian states characterized by the
covariance matrix Eq.~(\ref{sqthcorm}) with $c_1=-c_2$.

In the present work we adopt the convention that STSs describe a physically rather frequent
situation in which the thermal noise acts possibly non-symmetrically on the two modes, that is $N_{th_1}\neq N_{th_2}$, and thus the total number of thermal photons is $N_{th_1}+N_{th_2}$. Here $r$ is the two-mode squeezing parameter and $a_{i}$ are the
annihilation operators in each of the two modes ($i=1,2$). The diagonal and off-diagonal covariance
matrix elements for these states, respectively $a=a_{sq-th}$, $b=b_{sq-th}$ and $c_1=-c_2=c_{sq-th}$, read:
where
\begin{eqnarray}
a_{sq-th}\! &\!=\!&\!\cosh (2r)\!+\!2N_{th_{1}}\cosh ^{2}(r)\!+\!2N_{th_{2}}\sinh ^{2}(r)\nonumber  \, , \\
b_{sq-th}\! &\!=\!&\!\cosh (2r)\!+\!2N_{th_{2}}\cosh ^{2}(r)\!+\!2N_{th_{1}}\sinh ^{2}(r) \, ,\nonumber \\
c_{sq-th}\! &\!=\!&\!(1+N_{th_{1}}+N_{th_{2}})\sinh (2r)\label{ccvsqth} \, .
\end{eqnarray}
where $N_{s}=\sinh^{2}{(r)}$ is the number of squeezed photons.

Thermal squeezed states (TSSs) describe the reverse physical situation: an initially two-mode squeezed vacuum is allowed to evolve
at later times in a noisy channel and eventually thermalizes with an external environment characterized by a total number of thermal
photons $N_{th_1}+N_{th_2}$. The covariance matrix elements of TSSs, respectively $a=a_{th-sq}$, $b=b_{th-sq}$ and $c_1=-c_2=c_{th-sq}$, are:
\begin{eqnarray}
a_{th-sq} &=&2N_{s}+1+2N_{th_1}\, , \nonumber \\
b_{th-sq} &=&2N_{s}+1+2N_{th_2}\, , \nonumber \\
c_{th-sq} &=&2\sqrt{N_{s}(N_{s}+1)} \; . \label{cthsq}
\end{eqnarray}
The same covariance matrix, Eq.~(\ref{sqthcorm}), also describes classical uncorrelated
tensor product states, which we assume to be Gaussian. Thermal states are obtained letting
$c=c_{cl}=0$, $a=a_{cl}=1+2N_{th_1}$ and $b=b_{cl}=1+2N_{th_2}$. These Gaussian states are classical in the
sense that they can be written as convex combinations of coherent states
and, moreover, they are the only Gaussian states with vanishing discord~\cite{Adesso2010,Adesso2011}.
Notice that in the standard quantum optics terminology the wording classical states is used to denote any state with positive Glauber-Sudarshan $P$-representation. In the following, without loss of generality, we will identify party $A$ with mode $a_{1}$ and party $B$ with mode $a_{2}$.

\subsection{Probability of error: upper and lower bounds, and Gaussian discords of response}
\label{GausstransA}

For unitary-coding protocols with Gaussian transmitters, Gaussian local (single-mode) unitary operations acting on an infinite-dimensional Hilbert space are implemented by local (single-mode) symplectic transformations acting on the covariance matrix $\sigma$ of two-mode Gaussian input states $\rho_{AB}^{(\sigma)}$. In the following we will consider only traceless transformations. The traceless condition must be imposed in order to exclude trivial coding by two identical channels, for which the maximum probability of error is always $1/2$. Moreover, imposing tracelessness allows to investigate and determine the correspondence between reading efficiency and quantum correlations, as will be shown in the following.
Denoting by $F_A$ the local traceless symplectic transformations acting on mode $A$, the two local unitary operations implementing the encodings of the binary memory cells are the identity $\idty_{A}\oplus\idty_{B}$ and $F_A\oplus\idty_B$.

In order to assess the performance of quantum and classical Gaussian resources in the unitary-coding quantum reading protocol we need to evaluate the upper and lower bounds, Eqs.~(\ref{QCB}) and (\ref{errorineq}), on the maximum probability of error $P_{err}^{(\max )}$, Eqs.~(\ref{PerrMax}) and (\ref{perrdiscord}), for Gaussian two-mode transmitters $\rho_{AB}^{(\sigma)}$. To this end, we introduce first the {\em Gaussian} discord of response~\cite{Buono2014}, i.e the discord of response obtained by minimizing over local unitaries restricted only to the subset of local symplectic, traceless, transformations $F_A$:
\begin{equation}
{\cal{GD}}_{R}^{x}(\rho_{AB}^{(\sigma)}) \equiv \min_{\{F_A\}} {\cal{N}}_x^{-1} d_{x}^2\left(\rho_{AB}^{(\sigma)},\widetilde{\rho }_{AB}^{(\sigma)}
\right) \, ,
\label{intdr}
\end{equation}
where the index $d_x$ stands for trace, Hellinger, or Bures distance with the same normalization factors ${\cal{N}}_x^{-1}$ as before, and $F_A^T$ is the transpose of the symplectic matrix $F_A$ and $\widetilde{\rho }_{AB}^{(\sigma)} \equiv \rho_{AB}^{(F_A \sigma F_A^T)}$.
The Gaussian discord of response provides an upper bound to the true discord of response of Gaussian states and vanishes on and only on Gaussian classical states (subset of separable states that are in product form). The main properties of the Gaussian discord of response are reported in Appendix~\ref{appdiscord}.

In complete analogy with Eq.~(\ref{perrdiscord}) the maximum probability of error in discriminating two Gaussian transmitters related by a local symplectic transformation can be expressed as a simple function of the trace Gaussian discord of response:
\begin{equation}
P_{err}^{(\max )} = \frac{1}{2} - \frac{1}{2}\sqrt{{\mathcal{GD}}_{R}^{Tr}\left(\rho_{AB}^{(\sigma)}\right)} \; \, .
\label{Gaussperrdiscord}
\end{equation}
Specializing the bounds given by Eq.~(\ref{errorineq}) to the maximum probability of error in distinguishing Gaussian states, one has:
\begin{equation}
LBP_{err}^{(\max )} \leq P_{err}^{(\max )} \leq QCB^{(\max )} \, ,
\label{Gaussbounds}
\end{equation}
where the lower bound $LBP$ is a simple monotonically non-increasing function of the Bures Gaussian discord of response:
\begin{eqnarray}
LBP_{err}^{(\max )} = \frac{1}{2}\left( 1-\sqrt{1 - \big( 1 - {\mathcal{GD}}_{R}^{Bu} \big)^{2}} \right)
\, , \label{LBPmax}
\end{eqnarray}
and the upper bound $QCB$ is a simple linear, monotonically non-increasing function of the Hellinger Gaussian discord of response:
\begin{eqnarray}
QCB^{(\max )} = \frac{1}{2}\left( 1 - {\mathcal{GD}}_{R}^{Hell} \right) \, .
\label{QCBmax}
\end{eqnarray}
Therefore, for increasing Gaussian discords of response the bounds on the probability of error decrease correspondingly. The explicit expressions of the quantum Chernoff bound $QCB$, the Hellinger Gaussian discord of response, the Uhlmann fidelity, and the Bures Gaussian
discord of response are derived in Appendices~\ref{UhlmannGaussian} and \ref{ChernoffGaussian}.

\subsection{Maximum probability of error:  $\pi/2$ phase shift}
\label{phaseshiftreading}

The probability of error in distinguishing $\rho_{AB}^{(\sigma)}$ from $\widetilde{\rho }_{AB}^{(\sigma)} \equiv \rho_{AB}^{(F_A \sigma F_A^T)}$ is given by Eq.~(\ref{proberr}) with the local symplectic transformations $F_A$ replacing $W_{A}$. Among the local unitary operations $F_A$ which can implement the unitary-coding reading protocol, an important subset includes the single-mode phase shifts $P_{\phi }$ acting on mode $a_{1}$, parameterized by the
angle parameter $\phi$: $P_{\phi }=\exp {(-i\phi a_{1}^{\dagger }a_{1})}$.

Under a local phase shift the local mode $a_{1}$ is transformed as follows: $\widetilde{a}_1 = P_{\phi }a_{1}P_{\phi }^{\dagger } = \exp {(-i\phi )}a_{1}$, while the two-mode covariance matrix $\sigma $ transforms according to
$(F_{\phi }\oplus \mathbbm{1})\sigma (F_{\phi } \oplus \mathbbm{1})^{T}$, where the symplectic matrix $F_{\phi }$ reads
\begin{equation}
F_{\phi }=%
\begin{bmatrix}
\cos {\phi } & \sin {\phi } \\
-\sin {\phi } & \cos {\phi }
\end{bmatrix} \, .
\end{equation}
For the maximum probability of error, Eq.~(\ref{Gaussperrdiscord}), the upper bound is achieved, from
Eqs.~(\ref{Gaussbounds}) and (\ref{QCBmax}), in terms of a simple linear function of the Hellinger Gaussian
discord of response. The latter, in turn, is obtained by minimizing the Hellinger distance over the entire set of
local unitary operations implemented on the covariance matrix by local symplectic, traceless, transformations.
For squeezed thermal and thermal squeezed states one finds that this minimum is realized by the $\pi/2$ phase shift $F_{\pi /2}$,
that is the only possible traceless phase shift. Therefore the extremal unitary-coding protocol in the ensemble of local
traceless symplectic operations is realized by a particular PSK coding, the phase shift $\pi/2$, which is the only traceless
PSK coding. The details of the proof are reported in Appendix~\ref{chbm}.

On the other hand the quantity $LBP_{err}^{(\max )}$, Eq.~(\ref{LBPmax}), evaluated at $\pi/2$,
may not be optimal but certainly still provides a lower bound on the maximum probability of error:
\begin{equation}
LBP_{err}(F_{\pi/2})\leq LBP_{err}^{max}\leq P_{err}^{max} \; .
\end{equation}

Since for a $\pi/2$ phase shift the corresponding transformation is implemented by the traceless symplectic matrix $F_{\pi/2}=\begin{bmatrix}0&1\\-1&0\end{bmatrix}$, the expectation values of the canonical quadrature operators $x$ and $p$
transform as follows: $\<x\> \rightarrow -\<p\>$ and $\<p\> \rightarrow \<x\>$. Therefore, undisplaced thermal
Gaussian states ($\<x\> = \<p\> = 0$) are left invariant, and the worst-case PSK coding $(\mathbbm{1},F_{\pi /2})$ is
completely invisible to classical transmitters (thermal states) since the $\pi /2$ shift
does not change their covariance matrix. The probability of error $P_{err}$ for every such
classical transmitter always achieves the absolute maximum $1/2$. Viceversa, the very same coding can always be
read by any quantum Gaussian transmitter with nonvanishing Gaussian discord of response. As a consequence, quantum transmitters
always outperform undisplaced classical transmitters in device-independent, worst-case scenario quantum reading.
The situation changes when we consider displaced thermal states, as displacement unavoidably increases distinguishability.
Indeed the coherent and thermal coherent states are very efficient in detecting phase shift transformations. Nonetheless, in
the next section we will show that thermal coherent transmitters are outperformed by noisy quantum ones provided that the
distribution of the thermal noise among the modes in the quantum resource is strongly {\em non-symmetric}.

\section{Comparing classical and quantum resources: noise-enhanced quantum transmitters}
\label{comparisonclass}

We have seen that without displacement classical transmitters (thermal states) are completely
blind to reading. Introducing displacement enhances the distinguishability
of output states and turns classical states (thermal coherent states) in useful transmitters.
It is straightforward to show that distinguishability and the reading efficiency increase by implementing a single-mode
displacement rather than a two-mode one with equal single-mode amplitudes.

Let us then consider a scenario in which one compares
discordant quantum transmitters with displaced classical ones. We will show that in this
case, that is comparing noisy quantum resources with distinguishability-enhanced noisy classical ones,
discordant transmitters can outperform classical ones, and that the quantum advantage increases
with increasing (thermal) noise.

Stated precisely, given the same coding $(\idty,F_{\pi /2})$ acting locally on the first mode, we want to identify the regimes in which the
probability of error associated to a quantum transmitter is smaller than the probability of error associated to a thermal coherent one.
From Eqs.~(\ref{QCB}) and (\ref{errorineq}) this is equivalent to identifying the regimes in which the upper bound $QCB$ on the
probability of error using squeezed thermal
transmitters, denoted by $QCB^{sq-th}$, is smaller than the lower bound $LBP_{err}$
using thermal coherent states, that will be denoted
by $LBP_{err}^{coh-th}$. Obviously, only a constrained comparison at given fixed
physical quantities is meaningful. We will thus compare squeezed
thermal states and displaced thermal states at
fixed purity and fixed total number of photons. We will observe that the quantum advantage is achieved provided
the covariance matrix is not symmetric with respect to exchange of the modes.

With these notations, the requirement for a \emph{bona fide} quantum advantage reads as follows:
\begin{equation}
QCB^{sq-th} \leq LBP_{err}^{coh-th} \; .
\label{uplow}
\end{equation}

Both the coherent thermal and squeezed thermal states are two extremal classes of the general family of states which can be described as squeezed displaced thermal states (SDTS), defined as:
\begin{equation}
\rho_{SDTS} = S(r)D(\alpha)\rho_{th}(N_{th_1},N_{th_2})D(\alpha)^{\dagger} S(r)^{\dagger} \; ,
\end{equation}
where $S(r)=\exp(r a_1^{\dagger}a_2^{\dagger}-r a_1a_2)$ is the two-mode squeezing operator and we assume that squeezing parameter $r$ is real. Here $D(\alpha)=\exp(\alpha a_1^{\dagger}-\bar{\alpha} a_1)$ is the single-mode displacement operator and $\rho_{th}(N_{th_1},N_{th_2})\equiv\rho_{th_1}\otimes \rho_{th_2}$ is the non-symmetric two-mode thermal state, where $\rho_{th_i}\equiv \frac{1}{1+N_{th_i}}\sum_{m=0}^\infty (\frac{N_{th_i}}{1+N_{th_i}})^m\ket{m_i}\bra{m_i}$.
The purity $\mu=1/(16\det\sigma)^{1/2}$ of the SDTS is a function of the covariance matrix $\sigma$ and
depends only on the number of thermal photons:
\begin{equation}
\mu=\frac{1}{(1+2N_{th_1})(1+2N_{th_2})} \; .
\end{equation}
The total number of photons, $N_T=<a_1^{\dagger}a_1+a_2^{\dagger} a_2>$, in the SDTS reads:
\begin{equation}
N_T = \left( N_{th_1} + N_{th_2}\right)\left( 1 + 2N_s \right) + 2N_s \left( 1 + |\alpha|^2 \right) + |\alpha|^2 \; ,
\end{equation}
where $N_s=\sinh(r)^2$ is the number of squeezed photons.

Putting $r=0$, SDTSs reduce to thermal coherent states $\rho(0,\alpha,N_{th_1},N_{th_2})$ with total number of photons $N_T=N_{th_1}+N_{th_2}+|\alpha|^2$. Decreasing the displacement amplitude $\alpha$ the distinguishability of coherent thermal transmitters is reduced. We want to investigate whether this loss of distinguishability can be compensated by the quantum contribution due to increase of $r$ keeping $N_T$ and the purity fixed. In the limiting situation when $\alpha=0$ the corresponding quantum state $\rho(r,0,N_{th_1},N_{th_2})$ is a squeezed thermal state (STS). In the following we will show that for STSs $\rho(r,0,N_{th_1},N_{th_2})$ and thermal coherent states $\rho(0,\alpha,N_{th_1},N_{th_2})$ with equal total number of photons $N_T$, Ineq.~(\ref{uplow}) is satisfied for some ranges of $N_{th_1}$ and $N_{th_2}$. The condition of equal total number of photons $N_T$
implies $|\alpha|^2 = 2\sinh(r)^2(1+N_{th_1}+N_{th_2})$.

In order to evaluate the Uhlmann fidelity ${\cal{F}}$ and the quantum Chernoff bound $QCB$ in Eq.~(\ref{uplow})
we need to know how the phase shift $F_{\pi/2}$ transforms the transmitters
that we wish to compare: the squeezed thermal states and the
thermal coherent states. The dependence of ${\cal{F}}$ and
$QCB$ on the displacement vector and on the covariance matrix of general
Gaussian states is reported in Appendices~\ref{UhlmannGaussian} and
\ref{ChernoffGaussian}. The Uhlmann fidelity providing the lower bound on
$P_{err}$ for thermal coherent states depends only on the displacement
vector, since the covariance matrix of thermal coherent states is unaffected by the
action of the symplectic transformation 
$\left(F_{\pi/2}\oplus\mathbbm{1}_B\right)\sigma\left( F_{\pi/2}\oplus
\mathbbm{1}_B\right)^T$, where $F_{\pi/2}\oplus\mathbbm{1}_B=
\begin{bmatrix}
0 & -1 \\
1 & 0
\end{bmatrix}
\oplus\mathbbm{1}_B$.
Without loss of generality, the displacement vector of a thermal coherent state can be written as
$\<u\>_{coh-th} = [\sqrt{2}\abs{\alpha},0,0,0]^T$. Under a $\pi/2$ phase shift the difference $\delta$ between the final and
the initial displacement vectors reads as follows:
\begin{equation}
\delta=F_{\pi/2}%
\begin{bmatrix}
\abs{\sqrt{2}\alpha} \\
0 \\
0 \\
0%
\end{bmatrix}%
-%
\begin{bmatrix}
\abs{\sqrt{2}\alpha} \\
0 \\
0 \\
0%
\end{bmatrix}%
=%
\begin{bmatrix}
\abs{\sqrt{2}\alpha} \\
-\abs{\sqrt{2}\alpha} \\
0 \\
0%
\end{bmatrix}
\; .
\end{equation}
The Uhlmann fidelity of a thermal coherent state is then
\begin{equation}
\c F^{coh-th} = \exp{\left(-\frac{2\abs{\alpha}^2}{\c A}\right)} \; ,
\end{equation}
where $\c A=(1+2N_{th_1})$. The $QCB$ of non-symmetric, undisplaced squeezed thermal
state depends only on the covariance matrix, Eq.~(\ref{sqthcorm}), with
entries Eqs.~(\ref{ccvsqth}), and its explicit
expression is reported in Appendix~\ref{ChernoffGaussian}.


In Fig.~\ref{boundsenergy}, upper panel, we report the exact values of $P_{err}^{sq-vac}$ for a squeezed vacuum with squeezing $r$, $P_{err}^{coh}$ for a coherent state $\ket{\alpha}$, and $P_{err}^{sq-coh}$ for a squeezed displaced vacuum with squeezing $r'$ and displacement $\beta$, in the absence of noise, $N_{th_1}=N_{th_2} = 0$, and at fixed total photon number $N_T = \abs{\alpha}^2 = 2\sinh^2{(r)}=\abs{\beta}^2(1+2\sinh^2{(r')})+2\sinh^2{(r')}$.
The coherent states outperform the quantum resources given by the undisplaced squeezed vacuum. The coherent transmitters are then compared with squeezed displaced vacua of the same energy. Even if the latter include a classical contribution due to displacement and a quantum contribution due to squeezing they are still outperformed by the classical coherent states. The quantum efficiency converges to the classical one in the high-energy limit. For completeness, in  Fig.~\ref{boundsenergy} we also report the quantum Chernoff bound $QCB^{sq-vac}$ for the squeezed-vacuum transmitters.
%
%

\begin{figure}[tbp]
\includegraphics[width=7.0cm]{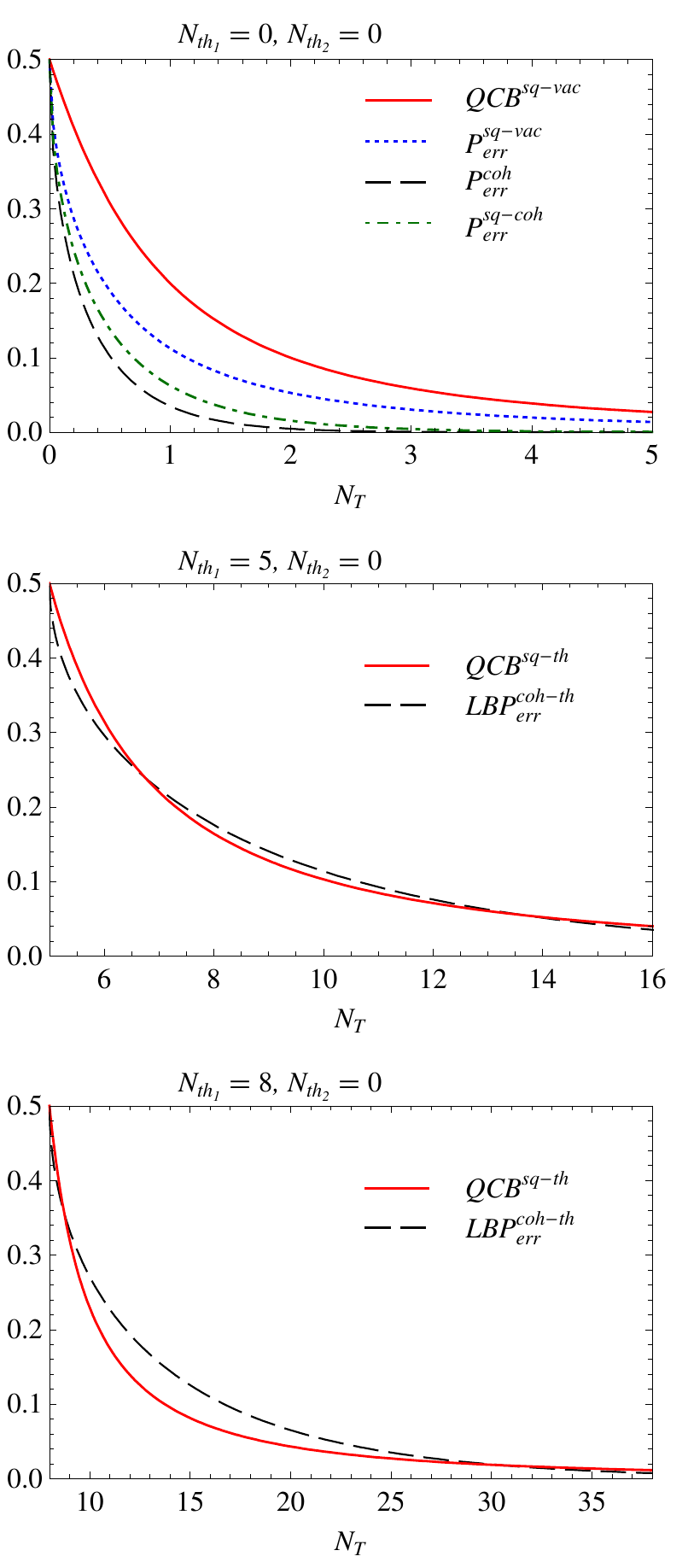}
\caption{Upper panel: behavior, as a function of the total photon number $N_T $, of the probability of error $P_{err}$ in the absence of thermal noise ($N_{th_1}=N_{th_2} = 0$). Blue dotted line: probability of error $P_{err}^{sq-vac}$ of
squeezed vacuum states. Black dashed line: $P_{err}^{coh}$ of
coherent states. Green dot-dashed line: $P_{err}^{sq-coh}$ of squeezed displaced vacuum states with displacement $\alpha=\frac{1}{2}N_{T}$. Red solid line: quantum Chernoff bound $QCB^{sq-vac}$ of squeezed vacuum states yielding the upper bound on $P_{err}$ with quantum transmitters.
No quantum gain is not observed in this regime.
Central panel: behavior of quantum and classical bounds on $P_{err}$ as
functions of $N_T$ at fixed asymmetric thermal noise: $N_{th_1} = 5$, $N_{th_2}=0$. Red
solid line: quantum upper bound $QCB^{sq-th}$ on $P_{err}$ with undisplaced
squeezed thermal states. Black dashed line: classical lower bound $LBP_{err}^{coh-th}$ on $P_{err}$
with thermal coherent states.
Lower panel: same as central panel, but with stronger thermal noise:
$N_{th_1} = 8$, $N_{th_2}=0$. With increasing $N_T$ the quantum upper bound goes below the
corresponding classical figures of merit and quantum transmitters certainly outperform classical ones.}
\label{boundsenergy}
\end{figure}

In the presence of symmetric thermal noise, $N_{th_1} = N_{th_2}$, there is no improvement in the quantum efficiency relative to the classical one. Introducing non-symmetric thermal noise, e.g. $N_{th_1}> N_{th_2}$, the Gaussian discords of response, that are intrinsically asymmetric quantities with respects to the subsystems in a given bipartition, increase dramatically, and so does the corresponding quantum reading efficiency. As a consequence, for sufficiently strong non-symmetric thermal noise the quantum resources outperform the classical ones. In the presence of non-symmetric noise exact expressions for $P_{err}$ are no
longer available. Therefore, in the central and lower panels of Fig.~\ref{boundsenergy} we report the exact lower and upper bounds on $P_{err}$ based
on the Uhlmann fidelity ${\cal{F}}$ and on the quantum Chernoff bound $QCB$. We observe that at intermediate
values of the total number of photons $N_T$ the quantum upper bound $QCB^{sq-th}$ on $P_{err}$ is strictly lower than the classical lower
bound $LBP_{err}^{coh-th}$, assuring that the quantum resources outperform the classical ones.
The classical transmitters (thermal coherent states) recover the quantum efficiency for large values of the total photon number.

Moreover, comparing the central and the lower panels in Fig.~\ref{boundsenergy}, we observe that as the number of thermal photons $N_{th_1}$ is
increased, the range of values of the total photon number $N_T$ for which
one has a quantum advantage increases.

In Fig.~\ref{rconstans} we provide a plot of the contour lines for the
differences $QCB^{sq-th} - LBP_{err}^{coh-th}$ for different asymmetries: $N_{th_2}=0$ (upper panel) and $N_{th_2}=0.5$ (lower panel) as functions of the total
photon number $N_T$ and of the purity (or, equivalently of the number of
thermal photons $N_{th_1}$). When these differences become negative, Ineq.~(\ref{uplow}) is satisfied and the quantum resources certainly outperform the
classical ones.

From the upper panel of Fig.~\ref{rconstans}, for $N_{th_2}=0$ comparing noisy quantum
transmitters with noisy coherent ones, one observes that $QCB^{sq-th} -
LBP_{err}^{coh-th} < 0$ in a large region of parameters. Fixing the
squeezing, so that the change in the total photon number $N_T$ is due only
to the change in the number of thermal photons $N_{th_1}$, corresponds to a
straight line in the plane (in the figure, drawn at $r=0.8$). Remarkably,
for these iso-squeezed states the quantum advantage increases with
increasing number of thermal photons. This is an instance of noise-enhanced
quantum efficiency that will be discussed further in Sec.~\ref{comparisonquant}.

In the lower panel of Fig.~\ref{rconstans} we decrease the asymmetry $(N_{th_2}=0.5)$.
We observe that the quantum gain is also achieved but in the range of much higher $N_{T}$.
Again, fixing the squeezing, e.g. at $r=0.8$, we notice that the quantum advantage increases with
thermal noise.

We remark that these results are obtained in a 
scenario in which we compare the minimum quantum efficiency (upper bound on the
error probability using quantum transmitters) with the maximum classical
efficiency (lower bound on the error probability using coherent thermal
transmitters). Therefore the actual quantum advantage will be even larger.
\begin{figure}[!th]
\includegraphics[width=7.0cm]{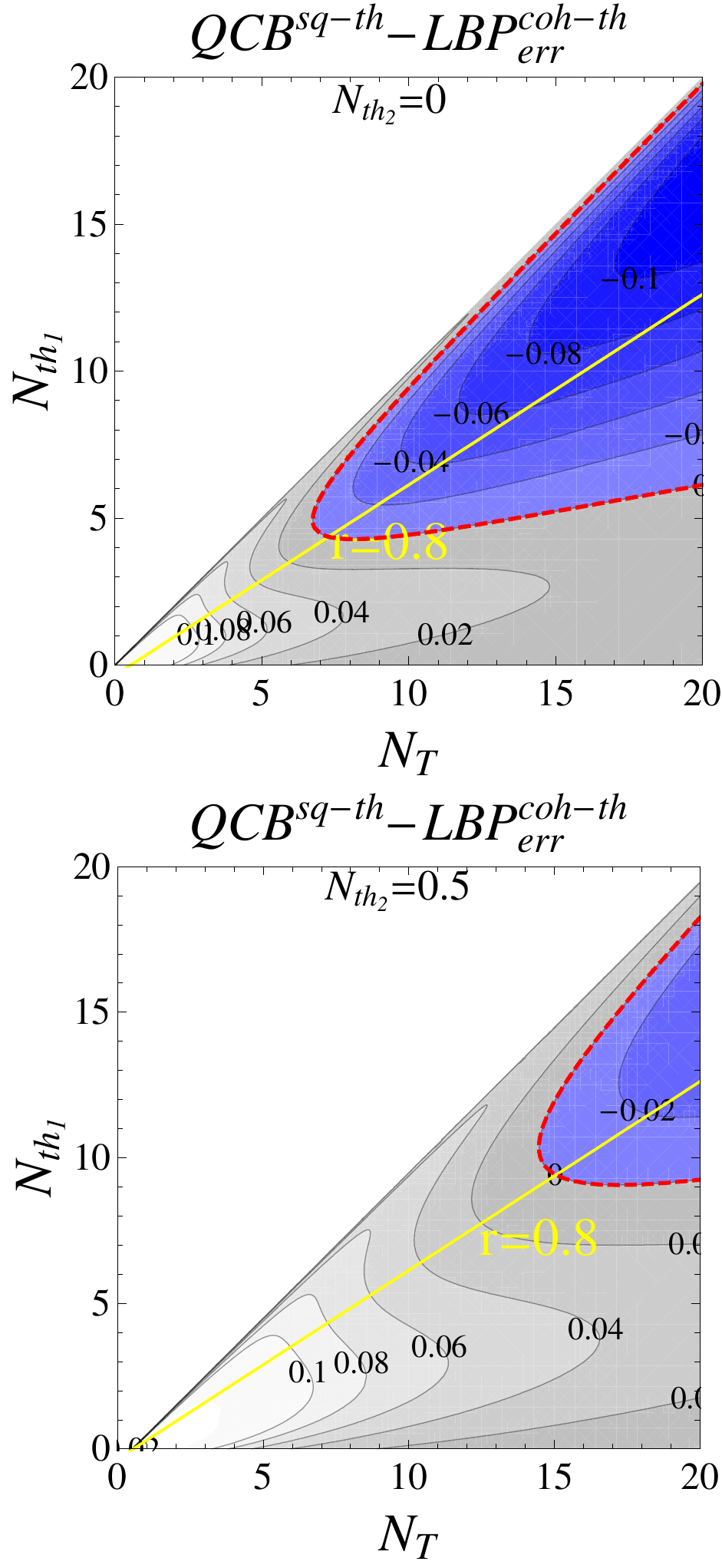}
\caption{Contour plot providing the contour lines for the differences $
QCB^{sq-th} - LBP_{err}^{coh-th}$ for $N_{th_2}=0$ (upper panel) and for $N_{th_2}=0.5$
(lower panel) as functions of the total photon number $N_T $ and of the number of thermal photons $N_{th_1}$. The region in which
these quantities assume negative values corresponds to quantum transmitters
certainly outperforming coherent thermal ones. The red dashed curve identifies its
boundary. The straight solid yellow lines in both panels corresponds to a fixed degree
of squeezing $r=0.8$. Moving along this lines in the direction of increasing
number of thermal photons $N_{th_1}$ one observes that as noise grows there is
a growing advantage in using quantum transmitters over classical ones. This
behavior provides an instance of noise-enhanced quantum performance. The effect is reduced when the asymmetry $N_{th_1}-N_{th_2}$ decreases.}
\label{rconstans}
\end{figure}

The quantum advantage disappears in the symmetric situation $N_{th_1}=N_{th_2}$.
Indeed, the inequality $N_{th_1} > N_{th_2}$ expresses the condition that the mode which passes through the coding channels is more noisy. This condition is unfavorable for thermal coherent states and favorable for STSs. Namely, in STSs with fixed finite squeezing, increasing the number of thermal photons in the first mode certainly increases the discord of response and, as a consequence, increases also the reading efficiency for this type of transmitters. This phenomenon is further analyzed in the following sections on the comparison of different quantum transmitters. These two concatenated effects cause the advantage of quantum states over the classical transmitters in the protocol of quantum reading with noisy transmitters. The asymmetry between the local thermal noise terms is the crucial element for realizing the enhancement of the reading efficiency. As we will see in the next section, the behavior of STSs with increasing number of thermal photons in the symmetric situation $N_{th_1}=N_{th_2}$, although not sufficient to realize a quantum advantage over classical resources, favors STSs among other noisy quantum transmitters.


\section{Comparing noisy quantum resources}
\label{comparisonquant}

In the previous section we compared classical and quantum transmitters and
for worst-case scenario we identified the regimes in which noisy but discordant quantum resources
outperform classical thermal coherent ones. We also observed that the quantum advantage
can increase, at fixed squeezing, with increasing thermal noise. We will now
compare the behavior of squeezed-thermal and thermal-squeezed states
in order to investigate how thermal noise affects the quantum efficiency
of different classes of quantum transmitters. We shall compare symmetric squeezed thermal
and thermal squeezed transmitters either at fixed number of thermal photons or at fixed squeezing.
We will then consider how non-symmetric noise further enhances the quantum efficiency by suppressing
the upper bound on the probability of error. Finally, we will investigate how the quantum efficiency
of different quantum transmitters improves when multiple reading operations are implemented at
fixed thermal noise.

\subsection{Comparing symmetric squeezed thermal and thermal squeezed transmitters: fixed noise}
\label{quantumthermal}

Let us start by comparing quantum reading with symmetric squeezed thermal and symmetric thermal squeezed
transmitters at fixed number of thermal photons and its performance as a function of the total number of photons. This comparison
is motivated by the fact that the interplay between quantum and thermal fluctuations
is very different for these two classes of quantum states. Squeezed thermal states (STSs) are
obtained by applying on thermal states, namely states that have already
thermalized (e.g. at the output of a noisy channel) a purely quantum operation, two-mode squeezing, that can be interpreted
as a re-quantization of the thermal vacuum. Viceversa, thermal squeezed states (TSSs) are
realized by letting pure squeezed vacua evolve and eventually thermalize in a noisy channel.

Both squeezed thermal and thermal squeezed states are two extremal classes of the very general family of squeezed thermal squeezed displaced states (STSDSs) which are defined as follows:
\begin{eqnarray}
&& \rho(r,N_{th_1},N_{th_2},r',\alpha) = \\
&&S(r)\Phi_{N_{th_1},N_{th_2}}\Big[S(r')D(\alpha)\rho_{vac}D(\alpha)^{\dagger}S(r')^{\dagger}\Big] S(r)^{\dagger} \, .
\nonumber
\end{eqnarray}
Here $S(r)$ and $S(r')$ are two-mode squeezing operators with different squeezing parameters $r$ and $r'$,
$D(\alpha)$ is a single-mode displacement operator, $\Phi_{N_{th_1},N_{th_2}}$ is a noisy channel introducing $N_{th_1}$ and $N_{th_2}$
thermal photons respectively in the first and in the second mode. The channel acts on a given Gaussian state adding the number of thermal photons to the diagonal entries of its covariance matrix. Finally, $\rho_{vac}=\ket{00}\bra{00}$ denotes the two-mode vacuum state.
We study this family of states at constant fixed values of the parameters $N_{th_1}$ and $N_{th_2}$.
The total number of photons in a STSDS is:
\begin{eqnarray}
&&N_T=<a_1^{\dagger}a_1+a_2^{\dagger} a_2>=\\
&&(N_{th_1}\! +\! N_{th_2}) \cosh(2 r')\! +\! (1 \! + \! |\alpha|^2) \cosh(2 (r \! + \! r')) \! - \!1 \, .
\nonumber
\end{eqnarray}
Consider first the situation without displacement, $\alpha=0$, and with symmetric thermal noise $N_{th_1}=N_{th_2}=N_{th}$.
Putting $r=0$, STSDSs reduces to thermal squeezed states TSTs $\rho(0,N_{th},N_{th},r',0)$.
Decreasing $r'$ and correspondingly increasing $r$ while keeping $N_T$ and $N_{th}$ fixed, in the limit $r' \rightarrow 0$ one recovers
the squeezed thermal states STSs $\rho(r,N_{th},N_{th},0,0)$.
Let us compare these two extremal classes of quantum Gaussian transmitters that coincide for $N_{th}=0$ and differ for $N_{th} \neq 0$ or, in non symmetric situations, when either $N_{th_1} \neq 0$ and/or $N_{th_2} \neq 0$.

In Fig.~\ref{quantumquantum}, upper left panel, we observe that for nonvanishing but small number of thermal photons $N_{th}$ the upper bound on the probability of error $P_{err}$ using STSs still remains above the lower bound on $P_{err}$ using TSSs. By further increasing thermal noise, as shown in the upper right panel of Fig.~\ref{quantumquantum}, all bounds with STSs are below all bounds with TSSs and the STSs certainly outperform the TSSs. Due to the different effects of the noise in STSs and TSSs we observe the clear advantage of using STSs over TSSs in a quantum reading protocol.
The lower panels of Fig.~\ref{quantumquantum} show the comparison done for displaced thermal squeezed states $\rho(0,N_{th},N_{th},r',\alpha)$ and displaced squeezed thermal states $\rho(r,N_{th},N_{th},0,\alpha)$. At fixed total number of photons we observe that single-mode displacement always increases the reading efficiency for both classes of states, while decreasing the squeezing reduces and eventually wipes out the quantum advantage of STSs over TSSs. These two classes of states coincide in the limiting case $r=r'=0$ in which they both recover the classical thermal coherent states. As we have seen in the previous section, the advantage of STSs over classical states is recovered by considering non-symmetric thermal noise.

\begin{figure}[tbp]
\includegraphics[width=9.0cm]{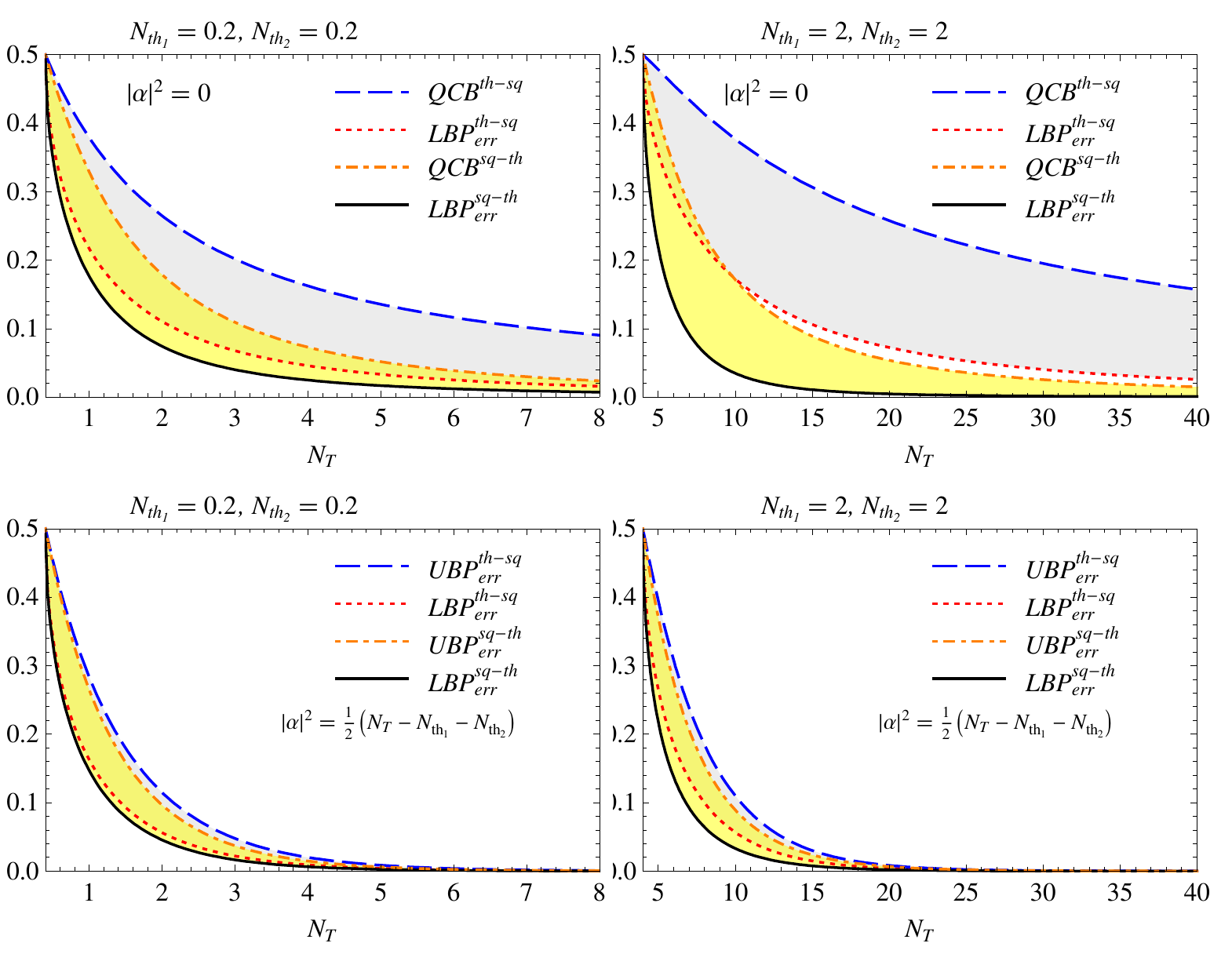}
\caption{Upper left panel: behavior, as a function of the total number of photons $N_T $, of the upper and lower bounds on the probability of error $P_{err}$ using either undisplaced squeezed thermal states (STSs) or undisplaced thermal squeezed states (TSSs) with fixed, symmetric, thermal noise: $N_{th_1} = N_{th_2}=0.2$. Upper right panel: the same but with $N_{th_1}=N_{th_2}=2$. In this case one observes that beyond a threshold value of $N_{T}$ the STSs certainly outperform the TSSs. Lower panels similar to the upper panels but with nonvanishing displacement $|\alpha|^2=\frac{1}{2}(N_T-N_{th_1}-N_{th_2})$. The upper bounds on the probability of error are given here by the quantum Bhattacharyya coefficient $UBP_{err}=1/2\tr\sqrt{\rho_1}\sqrt{\rho_2}$, which for states that include displacement does not necessary coincide with the quantum Chernoff bound $QCB$. The displacement increases the efficiency of the reading, however it does not guarantee with certainty the noise-enhanced performance of STSs with respect to TSSs.}
\label{quantumquantum}
\end{figure}

\subsection{Comparing symmetric squeezed thermal and thermal squeezed transmitters: fixed squeezing}

When comparing the behavior of the Uhlmann fidelity, quantum Chernoff bound, and
the Gaussian discords of response under variations of the classical noise at a fixed level
of quantum fluctuations (squeezing), we expect a radically diverging behaviors of the STSs with
respect to the TSSs. On intuitive grounds, since fidelity, Chernoff bound, and discord are measures of
distinguishability between an input state and the corresponding output after
a local disturbance, if we compare STSs and TSSs we notice from the structure of
their covariance matrices, see Eqs.~(\ref{cthsq}), that as
$N_{th}$ increases the correlation part of the STSs increases, while it remains constant in TSSs.

Indeed, the quantum Chernoff bound $QCB$ and of the Uhlmann fidelity $\c F$ for any two Gaussian
states of the form Eq.~(\ref{sqthcorm}) with $a=b$ and $c=c_1=-c_2$, related by a $\pi/2$ phase shift, take the form:
\begin{equation}
QCB = \frac{a^2-c^2}{2a^2-c^2} \; ,
\label{derivchernoff}
\end{equation}

\begin{equation}
\c F = \frac{4}{\left[1+c^2-a^2+\sqrt{(c^2-a^2)^2+1+2a^2}\right]^{2}} \; .
\label{derivfidel}
\end{equation}
In Fig.~\ref{e18} we report the behavior of the upper bound on the probability of error $QCB$, Eq.~(\ref{derivchernoff}), and of
the lower bound $LBP_{err}$ (which is a monotonic increasing function of the Uhlmann fidelity $\c F$, Eq.~(\ref{derivfidel}))
for the squeezed thermal and thermal squeezed states as functions of the number of thermal photons at fixed squeezing.
We observe that for TSSs $\c F$ and $QCB$ both increase with increasing thermal noise, converging asymptotically to
the absolute maximum ($1/2$) of the probability of error. Therefore, the quantum efficiency of TSSs is
suppressed by increasing the thermal noise.

On the contrary, for STSs $QCB$ remains constant
and $LBP_{err}$ decreases. This behavior guarantees that the probability of error, at fixed squeezing, is bound to vary in a
restricted interval below $0.1$. In the given example, the squeezing amplitude $r$ has been fixed at a relatively low value
$r \simeq 0.9$. Increasing the level of squeezing will further reduce the maximum value achievable by the probability of error.
In conclusion, the quantum advantage associated to squeezed thermal states is paramount at fixed, even moderate, squeezing,
and increases monotonically with increasing thermal noise.
\begin{figure}[tbp]
\includegraphics[width=7.9cm]{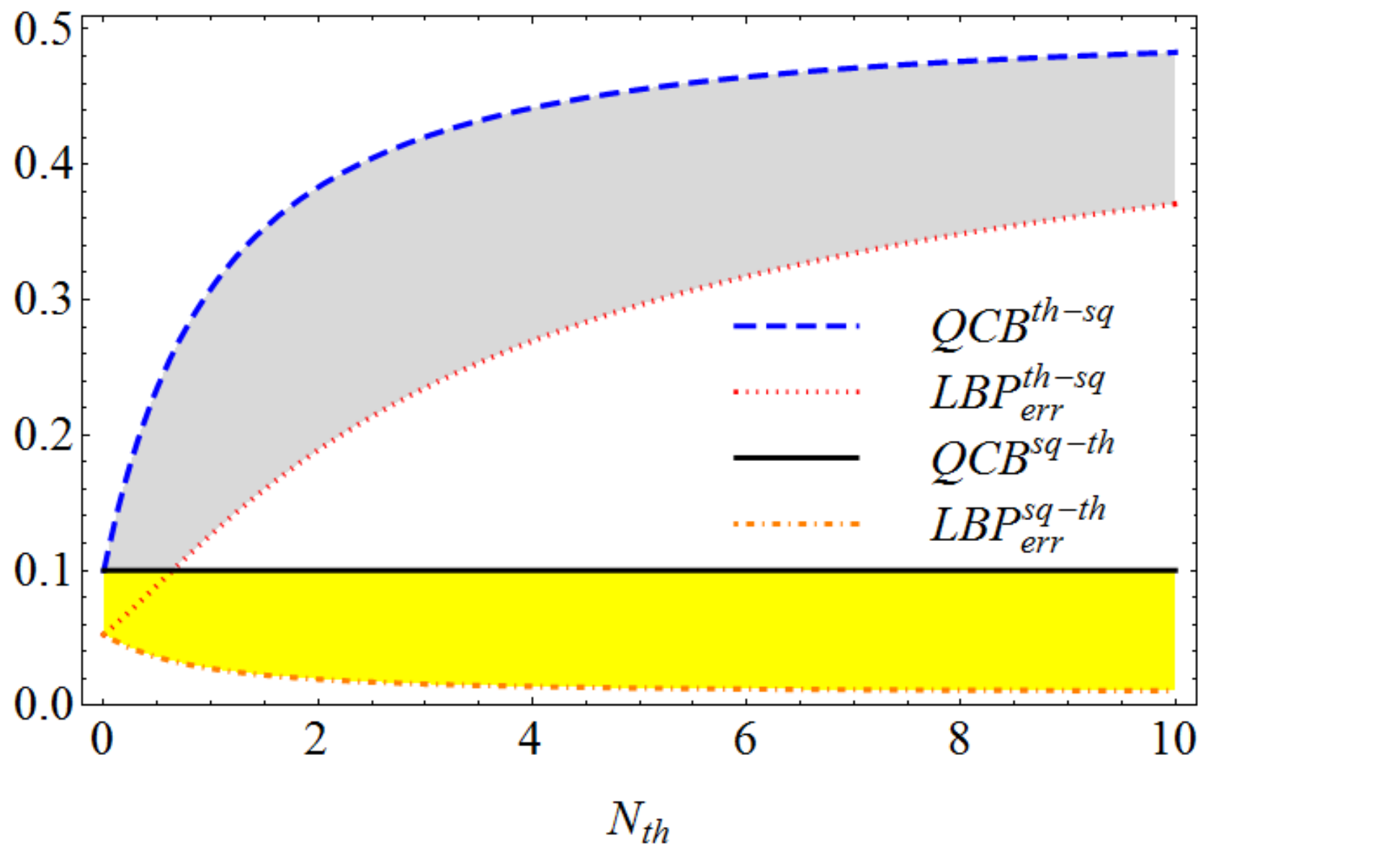}
\caption{Behavior of the quantum Chernoff bound $QCB$ and of the lower bound on the probability of error $LBP_{err}$
as functions of the number of thermal photons $N_{th}$, at fixed number of squeezed photons $N_s = 1$, for thermal squeezed
and squeezed thermal states. Blue dashed line: $QCB$ for thermal squeezed states. Dotted red line: $LBP_{err}$ for thermal
squeezed states. Solid black line: $QCB$ for squeezed thermal states. Orange dot-dashed line: $LBP_{err}$ for squeezed thermal
states. The colored areas between the upper and lower bounds denote the admissible intervals of variation for the probability
of error $P_{err}$. Increasing thermal noise suppresses the efficiency of thermal squeezed transmitters and increases the
efficiency of squeezed thermal ones.}
\label{e18}
\end{figure}
A more detailed understanding of these opposite behaviors can be gained by looking at the variation of the measures of
distinguishability with respect to the variations of the thermal noise and of the parameters of the covariance matrix.

Consider a generic measure of distinguishability denoted by $f(\rho_1,\rho_2)$ where $f$, among others, includes the Uhlmann
fidelity $\c F$ and the quantum Chernoff bound $QCB$. Consider then the {\em total} derivative of $f$ with respect to $N_{th}$,
keeping $r$ constant:
\begin{equation}
\frac{df}{dN_{th}}=\frac{\partial f}{\partial a}\bigg|_{c}\frac{\partial a}{\partial N_{th}}+\frac{\partial f}{\partial c}\bigg|_{a}\frac{\partial c}{\partial N_{th}} \; .
\label{derivative}
\end{equation}
Specializing to either $\c F$ or $QCB$ we obtain the explicit expressions of their derivatives, as reported in Appendix \ref{Derivatives}. From these explicit
expressions it follows that it is always
$$
\frac{\partial QCB}{\partial a}\big|_{c} \geq 0 \; , \; \frac{\partial QCB}{\partial c}\big|_{a} \leq 0 \; , \; \frac{\partial \c F}{\partial a}\big|_{c} \geq 0 \; , \; \frac{\partial \c F}{\partial c}\big|_{a} \leq 0 \; ,
$$
irrespective of the type of quantum transmitter considered.

Hence, if $f$ represents either the Uhlmann fidelity or the quantum Chernoff bound, the derivative $\frac{\partial f}{\partial a}\big|_{c}\geq 0$.
This behavior agrees with the intuition that the operation of increasing the diagonal entries of the covariance matrix and keeping the off-diagonal entries constant acts like a thermal channel which makes the initial state and the final state after the phase shift less distinguishable.
The behavior $\frac{\partial f}{\partial c}\big|_{a}\leq 0$ for both $\c F$ and $QCB$ is also intuitively clear, since changes in $\sigma$ under $F_{\pi/2}$ are the greater the larger the off-diagonal entries when keeping the diagonal $a$ constant.

Let us now discuss the state-dependent derivatives: for STSs and TSSs the partial derivatives  $\frac{\partial a}{\partial N_{th}}$
and $\frac{\partial c}{\partial N_{th}}$ are non-negative, therefore they cannot oppose the behavior of the state-independent part. For STSs they are given by $2\cosh{(2r)}$ and $2\sinh{(2r)}$ respectively, while for TSSs $\frac{\partial a}{\partial N_{th}}=2$ and $\frac{\partial c}{\partial N_{th}}=0$.
The behavior of $\cal F$ or $QCB$ with increasing $N_{th}$ depends then on the ratio of the positive and negative parts on the left hand side of Eq.~(\ref{derivative}).

For TSSs there is only a positive contribution in Eq.~(\ref{derivative}) and both $\c F$ and $QCB$ increase with increasing number of thermal photons. As a consequence, both the lower and the upper bounds on the probability of error must increase, as observed in Fig.~\ref{e18}.
On the other hand, for STSs the negative contribution always prevails when considering the Uhlmann fidelity, while the positive and negative contributions always cancel exactly when considering the quantum Chernoff bound, leading to a constant upper bound on the probability of error, as observed in Fig.~\ref{e18}.

The constant behavior of $QCB$ as a function of thermal noise for STSs can be also seen directly from Eq.~(\ref{derivchernoff}). This equation can be rewritten straightforwardly only in terms of $a/c$. Indeed, this ratio for STSs does not depend on $N_{th}$.

In this section we have considered reading protocols with binary coding given by the identity and the phase shift $\pi/2$, and transmitters implemented by symmetric STSs. This is actually a worst-case scenario in two respects. On the one hand, the phase shift $\pi/2$ provides the worst possible coding among all traceless local symplectic operations (maximum probability of error, device-independent reading). On the other hand, the {\em symmetric} STSs provide the worst possible transmitters among general STSs.

Indeed, in the next subsection we will show that non-symmetric STSs provide much larger quantum efficiencies and even effectively suppress the probability of error.

\subsection{Non-symmetric squeezed thermal states: noise-suppressed bounds on the probability of error}
\label{noiseenhancedQCB}

One might speculate that the increment of the Bures discord of response for increasing thermal noise and the corresponding decrement of the lower bound on the probability of error are due to the particular relation with the Bures metrics induced by the Uhlmann fidelity.
However, this is not the case. We will now show that if one considers non-symmetric two-mode STSs then also the Hellinger discord of response increases under increasing local thermal noise and therefore the corresponding upper bound on the probability of error decreases as well.
This is a strong indication that the true probability of error decreases as well with increasing thermal noise and thus that the use of discordant, non-symmetric STSs yields an absolute advantage, even over the use of entangled pure states, namely two-mode squeezed vacua with the same amount of squeezing as in the corresponding STSs.

The covariance matrix of non-symmetric two-mode STSs is given in Eq.~(\ref{sqthcorm}) with the parameters given in Eqs.~(\ref{ccvsqth}).
The corresponding $QCB$ achieves its maximum for the $\pi/2$ phase shift, as proven in Appendix~\ref{chbm}, and the exact expression of $QCB$ for non-symmetric STSs related by a $\pi/2$ phase shift is:
\begin{equation}
QCB=\frac{ab-c^2}{2ab-c^2}\;.
\label{derivchernoffas}
\end{equation}

Let us consider the variation $\frac{d QCB}{dN_{th_1}}$ of the quantum Chernoff bound, at constant squeezing $r$ and constant number of thermal photons
$N_{th_2}$ in the second mode, whose analytical expression is provided in Appendix~\ref{Derivatives}. From this expression it is clear that there is a range of values of $N_{th_1}$ and $N_{th_2}$, namely $N_{th_1} > N_{th_2}$, for which $\frac{d QCB}{dN_{th_1}} < 0$. Therefore, in this regime $QCB$ decreases with increasing $N_{th_1}$. On the other hand, $QCB$ increases with increasing $N_{th_1}$ if $N_{th_1}<N_{th_2}$. Henceforth, in the symmetric situation $N_{th_1}=N_{th_2}=N_{th}$ the quantum Chernoff bound is maximum and constant, independent of $N_{th}$, as discussed in the previous section.

\begin{figure}
\includegraphics[width=9cm]{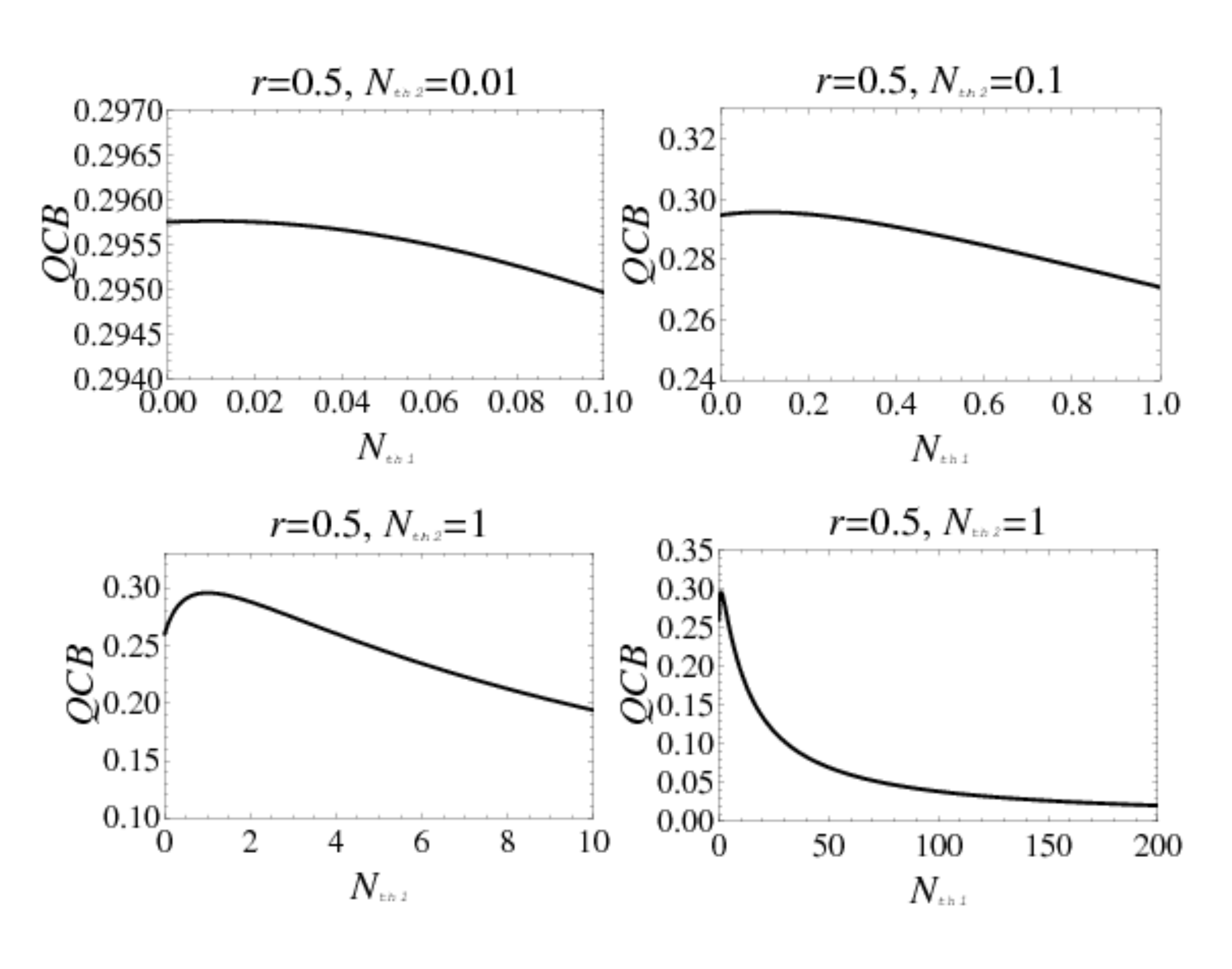}
\caption{Quantum Chernoff bound for non-symmetric STSs related by the phase shift $F_{\pi/2}$, as a function of $N_{th_1}$. In each panel the number of thermal photons $N_{th_2}$ in the second mode is fixed at a constant value. Upper left panel: $N_{th_2}=0.01$. Upper right panel: $N_{th_2}=0.1$. Lower left panel: $N_{th_2}=1$. Lower right panel: $N_{th_2}=1$ and extended range of values of $N_{th_1}$, in order to show the asymptotic vanishing of $QCB$ with increasing local thermal noise. For all panels the two-mode squeezing is fixed at $r=0.5$. The maximum of $QCB$ is achieved for symmetric STSs and provides the upper bound on the maximum probability of error $P_{err}^{\max}$ of the worst-case scenario.}
\label{qcbnth1}
\end{figure}

In Fig.~\ref{qcbnth1} we report the behavior of $QCB$ as a function of $N_{th_1}$ for different fixed values of $N_{th_2}$ and fixed squeezing $r$. In this physical situation the quantum  Chernoff bound decreases with increasing {\em local} thermal noise and vanishes asymptotically for $N_{th_1} \rightarrow \infty$. Therefore the probability of error in a Gaussian quantum reading protocol can be made arbitrarily small by using non-symmetric STSs transmitters with very large {\em local} thermal noise.

This very remarkable result may look at first quite counter-intuitive. In fact, the crucial point is that this feature is obtained by the {\em global quantum} operation of two-mode squeezing applied to a two-mode thermal state with very strong asymmetry in the {\em local} thermal noises affecting the two field modes. It is therefore not entirely unexpected that the consequences can be dramatic. While entanglement certainly decreases, the operation of squeezing a larger amount of noise can increase quantum state distinguishability by {\em "orthogonalizing"} on a larger portion of Hilbert space with respect to the thermal states.

\subsection{Squeezed thermal and squeezed vacuum states}
\label{noisypure}

Collecting all the previous results we are finally in the position to compare the best resources of device-independent Gaussian quantum reading, namely the noisy and discordant non-symmetric STSs, to the best absolute resources of Gaussian quantum reading, namely pure entangled two-mode squeezed vacuum states (TMSVSs). In the limit of infinite squeezing the TMSVSs are maximally entangled pure Einstein-Podolsky-Rosen (EPR) states whose probability of error in a quantum reading protocol vanishes identically. In absolute terms, TMSVSs are certainly the best among classical and quantum resources in a reading protocol with continuous variables. Indeed, in Fig.~\ref{sqthsqvac} we report the behavior of the exact probability of error for TMSVSs and the lower bound on it for non-symmetric STSs as functions of the total number of photons at fixed thermal noise, that is for {\em arbitrarily increasing squeezing} as the total number of photons increases. One observes that the lower bound on the probability of error for non-symmetric STSs is always above the exact probability of error for TMSVSs, converging towards it only asymptotically.

\begin{figure}
\includegraphics[width=8cm]{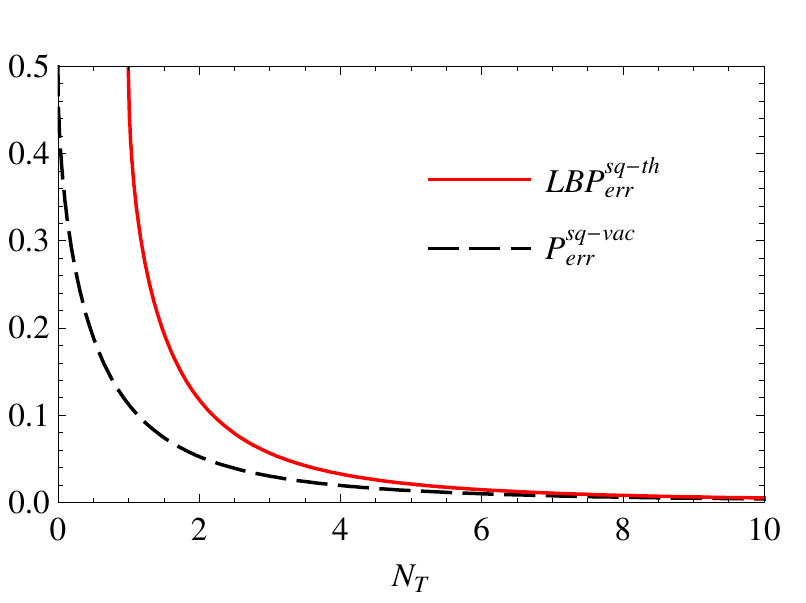}
\caption{Behavior as a function of the total number of photons $N_T$ of the probability of error $P_{err}^{sq-vac}$ for two-mode squeezed vacuum transmitters (TMSVs) and of the lower bound on the probability of error $LBP_{err}^{sq-th}$ for two-mode squeezed thermal transmitters (STSs) with $N_{th_1}=1$ and $N_{th_2}=0$. TMSVs have better reading efficiency than STSs. The squeezing in TMSVs is larger than the one in STSs at each fixed value of $N_T$. The two efficiencies converge asymptotically with increasing total number of photons.}
\label{sqthsqvac}
\end{figure}

On the other hand, it is also important to compare TMSVSs and non-symmetric STSs in terms of the concrete use of resources in realistically feasible experimental scenarios. In Fig.~\ref{stsvstmsvs} and Fig.~\ref{stsvstmsvs1} we report the behavior of the exact probability of error $P_{err}^{sq-vac}$ associated to TMSVS transmitters compared to the upper and lower bounds $QCB^{sq-th}$ and $LBP_{err}^{sq-th}$ for non-symmetric STS transmitters as functions of the total number of photons. In the case of TMSVSs the total number of photons obviously depends only on the squeezing and the behavior of the probability of error is the same as the one reported in Fig.~\ref{sqthsqvac}. However, at variance with Fig.~\ref{sqthsqvac}, in Fig.~\ref{stsvstmsvs} and \ref{stsvstmsvs1} we compare it with the lower and upper bounds for ST transmitters at a {\em fixed} finite value, low and comparably easy to produce experimentally, of the squeezing. In this case, the total number of photons in STSs varies only with the amount of thermal photons. Fig.~\ref{stsvstmsvs} shows the behavior of the exact probability of error for TMSVS transmitters and the bounds on the probability of error for non-symmetric STSs, as functions of the total number of photons $N_T$ and constant squeezing parameter fixed at $r=0.5$. For comparison, in Fig.~\ref{stsvstmsvs1} we report the same quantities but for a larger fixed two-mode squeezing $r=1$. Indeed, the higher the squeezing, the better the bounds on the probability of error using STS transmitters approximate the exact probability of error for TMSVS transmitters.

The crucial difference is that in real-world experimental setups it is comparatively much easier and less resource-demanding to implement a scheme relying on non-symmetric STSs with enhanced thermal noise and quantum discord than to produce pure (noise-free) TMSVSs with enhanced squeezing and entanglement.
Therefore, at {\em fixed} squeezing, we can compare the two classes of transmitters for {\em different} values of the total number of photons $N_T$ and ask for the threshold value of $N_{th_1}$ above which the discordant STSs certainly perform better than the entangled TMSVSs at the same fixed level of squeezing (the noise on the second mode being also fixed at a given reference value, say e.g. $N_{th_2}=0$). This threshold is thus determined by the condition $QCB^{sq-th} = P_{err}^{sq-vac}$. We give here two numerical examples for two different realistic values of the two-mode squeezing achievable in the laboratory with current technologies. For $r=0.5$, we have that $QCB^{sq-th} \leq P_{err}^{sq-vac}$ as soon as $N_{th_1} \geq 3.6$. For $r=1$, we have that $QCB^{sq-th} \leq P_{err}^{sq-vac}$ as soon as $N_{th_1} \geq 2.6$. Therefore, the higher the fixed level of squeezing, the lower is the level of thermal noise and quantum discord required to match the performance of pure entangled TMSVSs. Alternatively, increasing $N_{th_1}$ further above the threshold, we can also look for the complementary information on the minimum threshold values of $r$ (more easily realizable in the laboratory) above which STSs match or surpass the performance of TMSVSs at higher values of the squeezing (harder to achieve experimentally). In other words, we can introduce the concept of {\em effective} squeezing $r_{eff}$ associated to the value $N_{th}^{eff}(r,r_{eff})$ such that for $N_{th} > N_{th}^{eff}(r,r_{eff})$ STSs perform better than TMSVSs with given squeezing $r > r_{eff}$.

In conclusion, device-independent quantum reading is a remarkable protocol of quantum technology with noisy resources for which the best transmitters are discordant non-symmetric squeezed thermal states whose performance is optimized by realizing a fine trade-off between increased local thermal noise and fixed global two-mode squeezing, yielding noise-enhanced quantum correlations and state distinguishability.

\begin{figure}
\includegraphics[width=8cm]{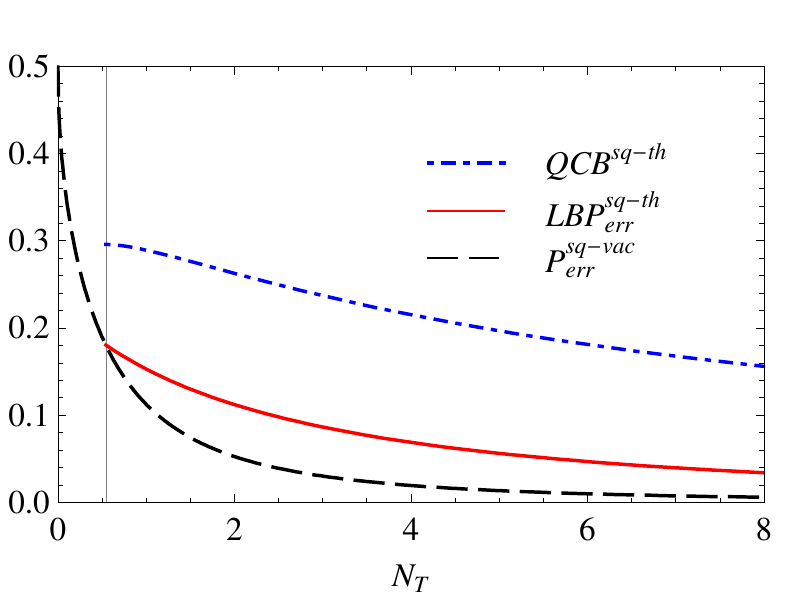}
\caption{Behavior as functions of the total number of photons $N_T$ of the probability of error $P_{err}^{sq-vac}$ for two-mode squeezed vacuum transmitters, and of the lower and upper bounds on the probability of error $LBP_{err}^{sq-th}$ and $QCB^{sq-th}$ for non-symmetric two-mode squeezed thermal transmitters. The latter two quantities are plotted for variable $N_{th_1}$ at {\em fixed} squeezing $r=0.5$, as well as fixed reference thermal noise in the second field mode $N_{th_2} = 0$.}
\label{stsvstmsvs}
\end{figure}

\begin{figure}
\includegraphics[width=8cm]{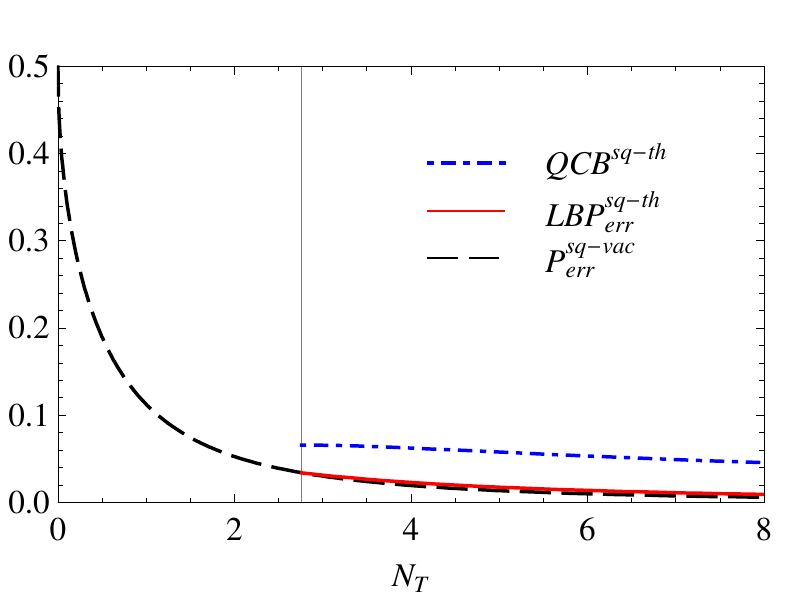}
\caption{Behavior as functions of the total number of photons $N_T$ of the probability of error $P_{err}^{sq-vac}$ for two-mode squeezed vacuum transmitters, and of the lower and upper bounds on the probability of error $LBP_{err}^{sq-th}$ and $QCB^{sq-th}$ for non-symmetric two-mode squeezed thermal transmitters. The latter two quantities are plotted for variable $N_{th_1}$ at {\em fixed} squeezing $r=1$, as well as fixed reference thermal noise in the second field mode $N_{th_2} = 0$.}
\label{stsvstmsvs1}
\end{figure}


\subsection{Many copies}\label{manycop}

Let us now analyze the case in which the total number of photons can vary by considering many copies of the transmitter,
that is repeating the reading protocol many times independently. Using $n$ copies of the system the Uhlmann fidelity
and the quantum Chernoff bound decrease as powers of $n$. Therefore, the probability of error can decrease both in the
case of squeezed thermal and thermal squeezed states.

The interesting question which arises here is how many copies we need in both cases to achieve a given level of probability of error.
The number of copies defines for instance the time needed for reading one bit of information in the given coding. Therefore this process
is interesting from the point of view of assessing the reading time and the strength of the sources of squeezed light that one needs.
\begin{figure}[!th]
\includegraphics[width=8cm]{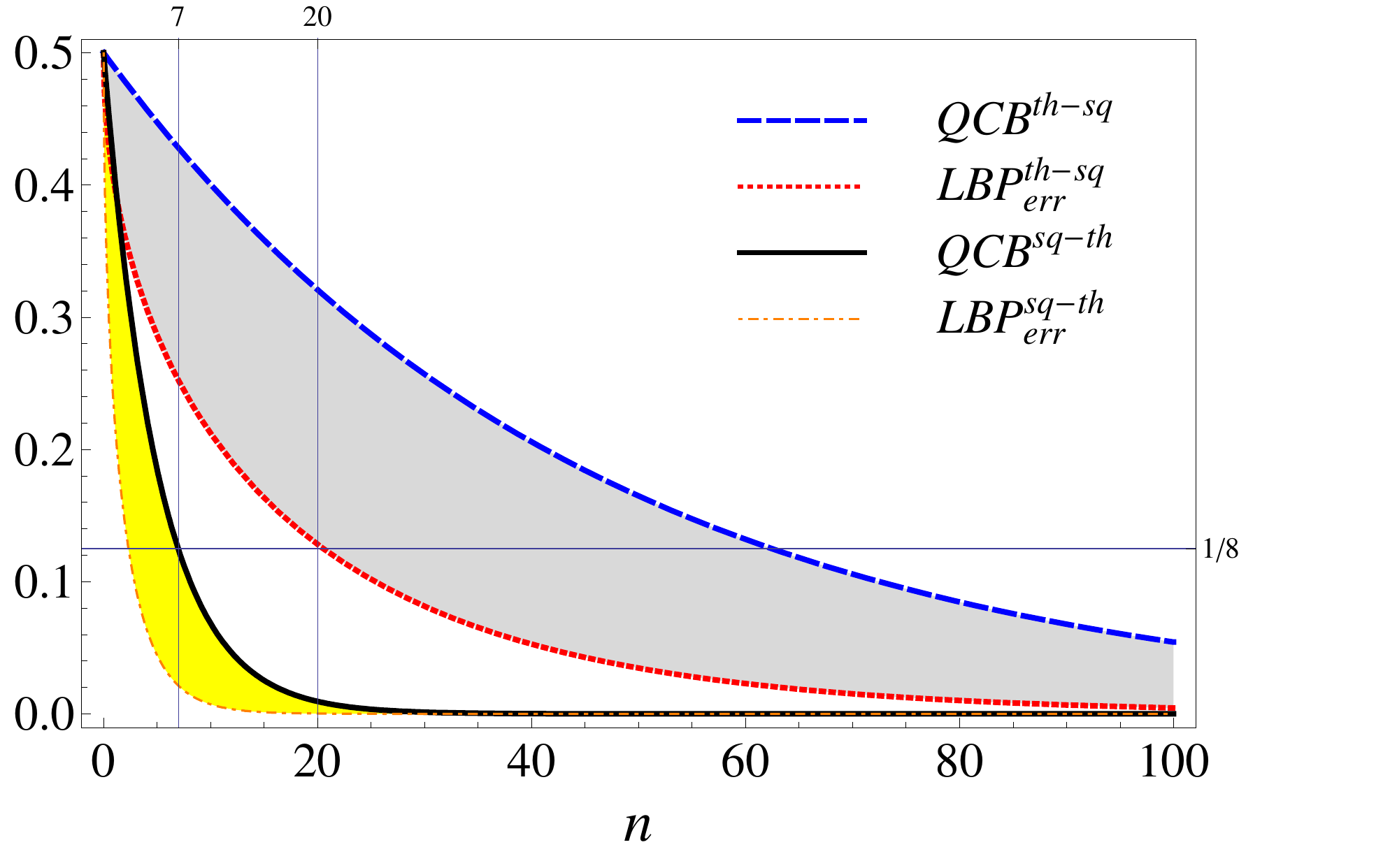}
\caption{Upper and lower bounds on the probability of error, using squeezed thermal (STSs) and thermal squeezed (TSSs) transmitters,
as a function of the number of copies of each transmitter, at fixed number of squeezed and thermal photons in each single
copy: $N_s=0.1$ and $N_{th}=1$. Blue dashed line: $QCB$ for thermal squeezed states. Dotted red line: $LBP_{err}$ for thermal
squeezed states. Solid black line: $QCB$ for squeezed thermal states. Orange dot-dashed line: $LBP_{err}$ for squeezed thermal
states. In order to achieve $P_{err}=1/8$ it is enough to take at most $n=7$ copies of STSs, while
the needed number of copies of TSSs is at least $n=20$.}
\label{en}
\end{figure}
Let us for instance assume that we require a value of the probability of error $1/8$, having for each copy of the squeezed thermal transmitter
the thermal noise fixed at $N_{th}=1$ and the weak squeezing fixed at $N_s=0.1$, see Fig.~\ref{en}. Looking at the upper bound (worst-case scenario),
the number of copies which are needed, in order to achieve the desired level of probability of error, is at most $n=7$.
Taking instead the thermal squeezed transmitter with the same squeezing and thermal noise in each copy, we see from Fig.~\ref{en}
that one needs, considering the lower bound (best-case scenario), at least $n=20$ copies.

These behaviors illustrate very clearly the advantage of using noise-enhanced quantum correlations. Indeed, comparing Figs.~\ref{en} and \ref{e18}, we see that by keeping a fixed level of squeezing and increasing the thermal noise, the number of copies of squeezed thermal transmitters needed to achieve a given level of precision stays constant, while the number of copies of thermal squeezed transmitters must increase.





%

\section{Conclusions and outlook}

\label{summary}

We have investigated Gaussian quantum reading protocols realized by weak optical sources in the worst-case scenario for quantum transmitters with respect to classical (thermal coherent) ones. For protocols that involve local unitary operations in the process of reading by continuous-variable Gaussian optical fields, we have showed that the maximum probability of error in reading binary memory cells is directly related to the amount of quantum correlations in a given transmitter, as quantified by the trace Gaussian discord of response. This relation allows to quantify the reading efficiency in terms of quantum correlations, providing a natural operational interpretation to the Gaussian discord of response.

Indeed, the latter is a well-defined measure of quantum state distinguishability under the action of local unitary operations. Therefore, the more discordant is the transmitter, the smaller is the maximum probability of error when using quantum resources. This relation then allows to determine the physical regimes of state purity and signal strength for which one has a net advantage in using quantum resources over classical thermal coherent ones.

Since the trace distance is in general uncomputable for Gaussian states, we have introduced exact upper and lower bounds on the maximum probability of error. We have showed that these bounds are expressed in terms of other type of quantum discords. In particular, the lower bound is expressed in terms of the Bures Gaussian discord of response, while the upper bound, provided by the quantum Chernoff bound maximized over the set of possible local unitary operations, is expressed in terms of the Hellinger Gaussian discord of response for squeezed thermal states and thermal squeezed states.

Both bounds decrease with an increasing amount of quantum correlations, providing a precise quantitative estimate of the quantum advantage obtained by using discordant resources over the corresponding thermal coherent ones. Moreover, the Bures and Hellinger discords of response are of further independent interest, as they play a central role in other quantum protocols studied recently, ranging from the assessment and use of local quantum uncertainty in optimal phase estimation~\cite{Girolami2013}, the efficiency of black-box quantum metrology~\cite{Interferometry2014,Spehner2013,GaussianMetrology2014}, and the quantum advantage of discordant resources in the protocol of quantum illumination~\cite{Farace2014}.

After comparing quantum and classical resources, we have discussed two fundamental classes of Gaussian quantum transmitters: symmetric squeezed thermal states (STSs) and symmetric thermal squeezed states (TSSs). We have shown that the actual beneficial or detrimental effects of environmental noise depend on the type of quantum state being considered. Considering STSs as quantum transmitters, the upper and lower bounds on the probability of error decrease with increasing thermal noise and therefore the quantum reading efficiency increases. The opposite behavior is observed when considering TSSs: in this case both the upper and the lower bounds on the maximum probability of error increase and therefore the quantum reading efficiency decreases with increasing thermal noise.

Finally, we went a step further and investigated the use of non-symmetric STSs. For such transmitters, also the quantum Chernoff bound decreases when the local thermal noise increases in one mode and remains fixed in the second mode. Indeed, the quantum Chernoff bound vanishes asymptotically with very large local thermal noise and therefore the probability of error must also vanish. In other words, non-symmetric two-mode STSs with imbalanced thermal noise between the two modes achieve an asymptotically vanishing probability of error for very large values of the noise imbalance. For such asymptotic states the Hellinger and Bures Gaussian discords of response attain their maximum value, and the quantum Chernoff bound and Uhlmann fidelity vanish. As a consequence, all upper and lower bounds on the probability of error vanish, the probability of error itself vanishes, and perfect reading is approached asymptotically.

Since the quantum reading efficiency of non-symmetric two-mode squeezed thermal states is a non-decreasing function of thermal noise, there is no evident advantage in using pure-state squeezed transmitters, the two-mode squeezed vacuum states, or low-noise ones over non-symmetric two-mode STSs with large noise imbalance between the field modes, {\em as long as} the squeezing is kept {\em fixed} at a realistic, finite constant value achievable in concrete experiments with currently available technology. Hence, noisy STSs transmitters {\em can} provide a better quantum efficiency at {\em fixed} two-mode squeezing, provided thermal noise (number of thermal photons) is enhanced beyond the threshold value above which the upper bound on the probability of error (quantum Chernoff bound) for STSs goes below the exact probability of error for TMSVSs with the same, fixed level of squeezing.

This remarkable phenomenon of noise-assisted quantum correlations and quantum efficiency is eventually due to the fact that quantum state distinguishability is intimately related to the concept of geometric quantum correlations, as measured by the discords of response, and the observation that the former can increase under increasing thermal noise. In particular, maximum {\em local} noise enhancement leads to maximum {\em global} enhancement of quantum correlations. In forthcoming studies we will provide a general characterization and quantification of noise-suppressed vs. noise-enhanced quantum correlations for different classes of quantum states~\cite{taming2014}, and we will investigate the relations between different types of quantum correlations according to states, metrics, and operations~\cite{RSI2014}.


\acknowledgments
F.I. acknowledges valuable discussions with Gerardo Adesso. The authors acknowledge financial support from the Italian
Ministry of Scientific and Technological Research under the PRIN 2010/2011 Research Fund, and from the EU FP7 STREP Projects iQIT,
G.A. No. 270843, and EQuaM, G.A. No. 323714.


\appendix

\section{Gaussian discord of response}

\label{appdiscord}
Here we discuss the Gaussian discord of response, given by Eq.~(\ref{intdr}), and prove that it is a {\it bona fide} measure of quantum correlations. More general discussion can be found in \cite{Buono2014}. The minimal set of axioms with universal consensus includes the following: $i)$ invariance under local unitary transformations, $ii)$ contractivity under the action of completely positive and trace preserving (CPTP) maps acting on mode $B$,  $iii)$ vanishing of quantum correlations if and only if the state is classically-quantum correlated, i.e. with block-diagonal covariance matrix, $iv)$ reduction to an entanglement monotone for pure states. The first condition is guaranteed by unitary invariance of the chosen distance and the procedure of minimization. The second condition is satisfied due to the fact that we consider only contractive distances in order to define the discord of response.

The third condition is verified as follows. It is known that classical-quantum two-mode
Gaussian states are those and only those which can be represented by the
tensor product $\omega_A\otimes\omega_B$ of single-mode Gaussian states~\cite{Adesso2010, Adesso2011}.
Up to displacements, such states are characterized by the block
diagonal covariance matrices $\sigma_{AB}^{(cq)}=\begin{pmatrix}
\sigma_A & 0 \\
0 & \sigma_B
\end{pmatrix}
$. Let us consider the local traceless symplectic transformation $F_A$ which can be decomposed as $F_A=S_AF_{\pi/2}S_A^{-1}$, where $F_{\pi/2}=\begin{pmatrix}
0 & 1 \\
-1 & 0
\end{pmatrix}
$ and $S_A$ is a symplectic matrix which diagonalize $\sigma_A$, i.e. $\sigma_A=\nu S_A\idty S_A^T$. Here $\nu$ are two equal symplectic eigenvalues of  $\sigma_A$.
The transformation given by $F_{\pi/2}$ is
symmetry-preserving and therefore $F_A=S_AF_{\pi/2}S_A^{-1}$ leaves $\sigma_A$ invariant.
This shows that if the state is classically correlated
there exists at least one
local traceless transformation $F_A$ that leaves the state invariant.

We now prove the reverse statement that only in the case of classically
correlated states there exists such a symplectic traceless transformation
that leaves the state invariant. Assume that the covariance matrix left invariant by a
traceless transformation $F_{A}$ has the form $\sigma_{AB}=
\begin{pmatrix}
L_{11} & L_{12} \\
L_{21} & L_{22}
\end{pmatrix}
$. Local symplectic transformation can bring the covariance matrix in the so called
normal form in which $L_{12}=
\begin{pmatrix}
c & 0 \\
0 & -c
\end{pmatrix}
$. If the state is not changed by the local transformation we have $F_AL_{12}=L_{12}$. Since $L_{12}$ is reversible we
obtain that $F_A=\mathbbm{1}$ which contradicts the assumption on the
spectrum of $F_A$. This shows that condition $iii)$ is satisfied.
Condition $iv)$ is guaranteed by the fact that for pure states the Gaussian discord of response reduces to the Gaussian entanglement of response~\cite{Adesso2007} (the Gaussian counterpart of the entanglement of response~\cite{Monras2011}) which is a
{\it bona fide} measure of entanglement.

\section{Uhlmann fidelity}

\label{UhlmannGaussian}

The Uhlmann fidelity for two-mode Gaussian states can be computed as follows~\cite{Marian2012}. Let us define the matrix of the symplectic form
\begin{equation}
\Omega=
\begin{bmatrix}
0 & 1 & 0 & 0 \\
-1 & 0 & 0 & 0 \\
0 & 0 & 0 & 1 \\
0 & 0 & -1 & 0
\end{bmatrix}
.
\end{equation}
The displacement vector is the vector of the averages of the amplitude and phase field quadratures $x$ and $p$ i.e. $\<u\>_{\rho}=(\<x_1\>,\<p_1\>,\<x_2\>,\<p_2\>)^T$, where $T$ stands for transposition. Denote the difference of the displacement vectors of two Gaussian states  $\rho_1$ and $\rho_2$ by  $\delta=\<u\>_{\rho_1}-\<u\>_{\rho_2}$.
We need also the auxiliary formulas defined using the covariance matrices $\sigma_1$ and  $\sigma_2$  of the respective Gaussian states:
\begin{eqnarray}
\Delta &=&\det(\sigma_1+\sigma_2), \\
\Gamma &=& 2^4 \det[(\Omega\sigma_1)(\Omega\sigma_2)-\frac{1}{4}\mathbbm{1}], \\
\Lambda &=& 2^4 \det(\sigma_1+\frac{i}{2}\Omega)\det(\sigma_2+\frac{i}{2}\Omega).
\end{eqnarray}
The Uhlmann fidelity for two mode Gaussian states is then
\begin{eqnarray}
\c F(\rho_1,\rho_2)&\equiv& \exp{\left[-\frac{1}{2}\delta^T(\sigma_1+
\sigma_2)^{-1}\delta\right]}  \notag \\
&\times&\left[(\sqrt{\Gamma}+\sqrt{\Lambda})-\sqrt{(\sqrt{\Gamma}+
\sqrt{\Lambda})^2-\Delta}\right]^{-1}.
\label{gfidelity}
\end{eqnarray}

\section{Quantum Chernoff bound}

\label{ChernoffGaussian}

Any $n$-mode Gaussian state can be represented in its normal mode decomposition
parameterized by $\rho\rightarrow(\<u\>,S,\{\nu_k\})$ in which $\<u\>$ is the vector of the averages of the quadratures and
\begin{equation}
\rho=U_{\<u\>,S}\left[\bigotimes_{k=1}^n\rho(\nu_k)\right]U_{\<u\>,S}^{\dagger},
\end{equation}
where
\begin{equation}
\rho(\nu_k)=\frac{2}{2\nu_k+1}\sum_{j=0}^{\infty}\left(\frac{2\nu_k-1}{
2\nu_k+1}\right)^j\left|{j_k}\right\rangle \left\langle{j_k}\right |
\end{equation}
is a thermal state with mean photon number $\bar{n}_k=\nu_k-1/2$ and $\left|{j_k}\right\rangle$ are
the eigenstates of the operator of the number of photons in mode $k$. The set $\{\nu_1,....,\nu_n\}$ identifies the symplectic spectrum. In this way the
covariance matrix is decomposed as
\begin{equation}
\sigma=S\tilde{\Lambda} S^T,\quad {\rm where}\quad \tilde{\Lambda}=\bigoplus_{k=1}^n\nu_k\mathbbm{1}_k.
\end{equation}
For two arbitrary Gaussian states with normal mode decompositions
$\rho_1\rightarrow (\<u_1\>,S_1,\{\alpha_k\})$ and $\rho_2\rightarrow (\<u_2\>,S_2,\{\beta_k\})$, assuming that $\delta=\<u_1\>-\<u_2\>$, we have~\cite{Pirandola2008}
\begin{eqnarray}
Q_t&\equiv&\tr{\rho_1^t\rho_2^{(1-t)}}\nonumber\\
&=&\bar{Q}_t\exp{\{-\frac{1}{2}\delta^T[V_1(t)+V_2(1-t)]^{-1}\delta\}},  \label{appchs}
\end{eqnarray}
where
\begin{equation}
\bar{Q}_t=\frac{2^n\prod_{k=1}^nG_t(\alpha_k)G_{1-t}(\beta_k)}{\sqrt{\det[V_1(t)+V_2(1-t)]}}
\label{qubar}
\end{equation}
and
\begin{equation}
G_p(x)=\frac{2^p}{(x+1)^p-(x-1)^p}.
\end{equation}
Moreover
\begin{eqnarray}
V_1(t)&=&S_1\left[\bigoplus_{k=1}^n\Lambda_t(\alpha_k)\mathbbm{1}_k\right]S_1^T,\label{v1} \\
V_2(1-t)&=&S_2\left[\bigoplus_{k=1}^n\Lambda_{1-t}(\beta_k)\mathbbm{1}_k\right]S_2^T,\label{v2}
\end{eqnarray}
where
\begin{equation}
\Lambda_p(x)=\frac{(x+1)^p+(x-1)^p}{(x+1)^p-(x-1)^p}
\end{equation}
The quantum Chernoff bound for arbitrary states $\rho_1$ and $\rho_2$ is
\begin{equation}
QCB\equiv\frac{1}{2}\inf_{t\in (0,1)}\tr{\rho_1^t\rho_2^{(1-t)}}
\label{appchernoff}
\end{equation}
which for Gaussian states is expressed by means of Eq.~(\ref{appchs}), i.e. $QCB=\frac{1}{2}\inf_{t\in (0,1)}Q_t$.

\section{Extremization of the quantum Chernoff bound}
\label{chbm}

Let us discuss the extremizations of the quantum Chernoff bound, $QCB$, between two states related by a local unitary transformation.
Lemma 1 in Ref.~\cite{Farace2014} shows that in the finite-dimensional case if $\rho _{2}=\Theta \rho _{1}\Theta ^{\dagger }$, where $\Theta $ is a
Hermitian matrix, the infimum is achieved for $t=1/2$ in Eq.~(\ref{appchernoff}). The same proof can be applied as well to $\Theta $ one-qubit traceless unitary matrix, since it is Hermitian. For Gaussian states of infinite-dimensional, continuous-variable systems, we are able to formulate and prove the following theorem:

\begin{theorem}
For two-mode Gaussian states $\rho^{(\sigma)}$ with covariance matrix $\sigma$ of the form Eq.~(\ref{sqthcorm}) with $c_1=-c_2$, the $QCB$ for the pair $(\rho^{(\sigma)},\rho^{(S_A\sigma S_A^T)})$, where $S_A$ is any traceless local symplectic transformation, is achieved for $t=1/2$, namely:
\begin{equation}
\tr\left(\rho^{(\sigma)}\right)^t\left(\rho^{(S_A\sigma S_A^T)}\right)^{1-t}\geq \tr\sqrt{\rho^{(\sigma)}}\sqrt{\rho^{(S_A\sigma S_A^T)}}.
\label{qcbs}
\end{equation}
\end{theorem}

\begin{proof}
First let us notice that for any two quantum states $\rho_1$ and $\rho_2$ function $\tr\rho_1^t\rho_2^{1-t}$ is convex in $t$, which is proven in~\cite{Audenaert2007}. We will show that if the two states $\rho_1$ and $\rho_2$ satisfy the assumptions of the theorem, $\tr\rho_1^t\rho_2^{1-t}$ is symmetric with respect to exchange $t\rightarrow 1-t$. These two properties imply the theorem.

Let us show the symmetry with respect to exchange $t\rightarrow 1-t$. From~\cite{Pirandola2008} we know that for the states with vanishing first moments   $\tr\rho_1^t\rho_2^{1-t}$ is given in Eq.~(\ref{qubar}). The numerator of this formula is already symmetric with respect to exchange $t\rightarrow 1-t$ for the states related by any unitary transformation. To proof the theorem we only need to show the symmetry of the determinant in denominator of Eq.~(\ref{qubar}), $\det\left(V_1(t)+V_2(1-t)\right)$, for the states which have the covariance matrices $V_1(t)=V(t)$ and $V_2(1-t)=S_AV(1-t)S_A^T$.

Consider the determinant from the above formula
\begin{eqnarray}
&&\det\left(V(t)+S_AV(1-t) S_A^T\right)\\
&=&\det\left(S_A^{-1}V(t) (S_A^T)^{-1}+V(1-t)\right)\\
&=&\det\left(S_AV(t) S_A^T+V(1-t)\right).
\label{sigmaz}
\end{eqnarray}
The last equality is implied by the following argument. The form of the most general single-mode traceless symplectic transformation (Euler decomposition)~\cite{Adesso2007} is:
\begin{equation}
S_A=
\begin{bmatrix}
\cos {\phi } & \sin {\phi } \\
-\sin {\phi } & \cos {\phi }
\end{bmatrix}
\begin{bmatrix}
\xi & 0 \\
0 & \xi ^{-1}
\end{bmatrix}
\begin{bmatrix}
\cos {\theta } & \sin {\theta } \\
-\sin {\theta } & \cos {\theta }
\end{bmatrix}
.  \label{apps}
\end{equation}
where $\xi $ is positive. The traceless condition is obtained by imposing $\phi =\pi /2-\theta $. Immediate check gives us that $S_A^{-1}=-S_A$. The minus sign is irrelevant in expression $S_A^{-1}V(t) (S_A^T)^{-1}=(-S_A)\otimes \idty V(t) (-S_A^T)\otimes \idty$ and can be omitted. This completes the proof.
\end{proof}

To express the upper bound on the maximum probability of error for undisplaced Gaussian states of the covariance matrix given in  Eq.~(\ref{sqthcorm}) with $c_1=-c_2$  and its counterpart related to it by traceless symplectic transformations $S_A$ we maximize the quantum Chernoff bound over the set of these transformations. Formula Eq.~(\ref{appchs}) with $t=1/2$ is maximized if $\det{[V(1/2)+S_AV(1/2)S_A^T]}$ is minimized over the set $\{S_A\}$. The most general Gaussian single-mode unitary transformation is given in Eq.~(\ref{apps}). The determinant $\det{[V(1/2)+S_AV(1/2)S_A^T]}$ does not depend on $\phi $ and achieves its minimum for $\xi =1$. This can be proved by direct check of the first and second derivatives. Substituting
$\xi =1$ in Eq.~(\ref{apps}) yields the transformation which maximizes the quantum Chernoff bound, namely
$$S_A=\begin{bmatrix}
0 & 1 \\
-1 & 0 \; ,
\end{bmatrix}$$
which is the transformation corresponding to a local $\pi/2$ phase shift $F_{\pi /2}$.

\section{Distinguishability measures and thermal noise}

\label{Derivatives}

Here we discuss the derivatives of the distinguishability functions and their behavior.

For the quantum Chernoff bound, $QCB$, the derivatives over the entries of the covariance matrix, Eq.~(\ref{sqthcorm}), read
\begin{eqnarray}
\frac{\partial QCB}{\partial a}\big|_{c}&=&\frac{2 a c^2}{\left(c^2-2 a^2\right)^2}\label{dqda}\;,\\
\frac{\partial QCB}{\partial c}\big|_{a}&=&-\frac{2 a^2 c}{\left(c^2-2 a^2\right)^2}\label{dqdc}\;.
\end{eqnarray}
It is immediate to verify that the partial derivative of $QCB$, Eq.~(\ref{dqda}), is always positive and the partial derivative, Eq.~(\ref{dqdc}), is always negative. The derivatives of the Uhlmann fidelity are
\begin{widetext}
\begin{eqnarray}
\frac{\partial \c F}{\partial a}\big|_{c}&=&\frac{16 a \left(-a^2+c^2+\sqrt{a^4-2 \left(c^2-1\right) a^2+c^4+1}-1\right)}{\sqrt{a^4-2 \left(c^2-1\right) a^2+c^4+1} \left(\sqrt{a^4-2 \left(c^2-1\right) a^2+c^4+1}-\sqrt{\left(-a^2+c^2+1\right)^2}\right)^3}\; ,\label{dfda}\\
\frac{\partial \c F}{\partial c}\big|_{a}&=&-\frac{8 \left(2 (a-c) c (a+c)-2 c \sqrt{a^4-2 \left(c^2-1\right) a^2+c^4+1}\right)}{\sqrt{a^4-2 \left(c^2-1\right) a^2+c^4+1}
   \left(\sqrt{\left(-a^2+c^2+1\right)^2}-\sqrt{a^4-2 \left(c^2-1\right) a^2+c^4+1}\right)^3}\;.\label{dfdc}
\end{eqnarray}
\end{widetext}
The behavior of these rather complicated functions is reported graphically in Fig.~\ref{dfdx}.

\begin{figure}[!th]
\includegraphics[width=5cm]{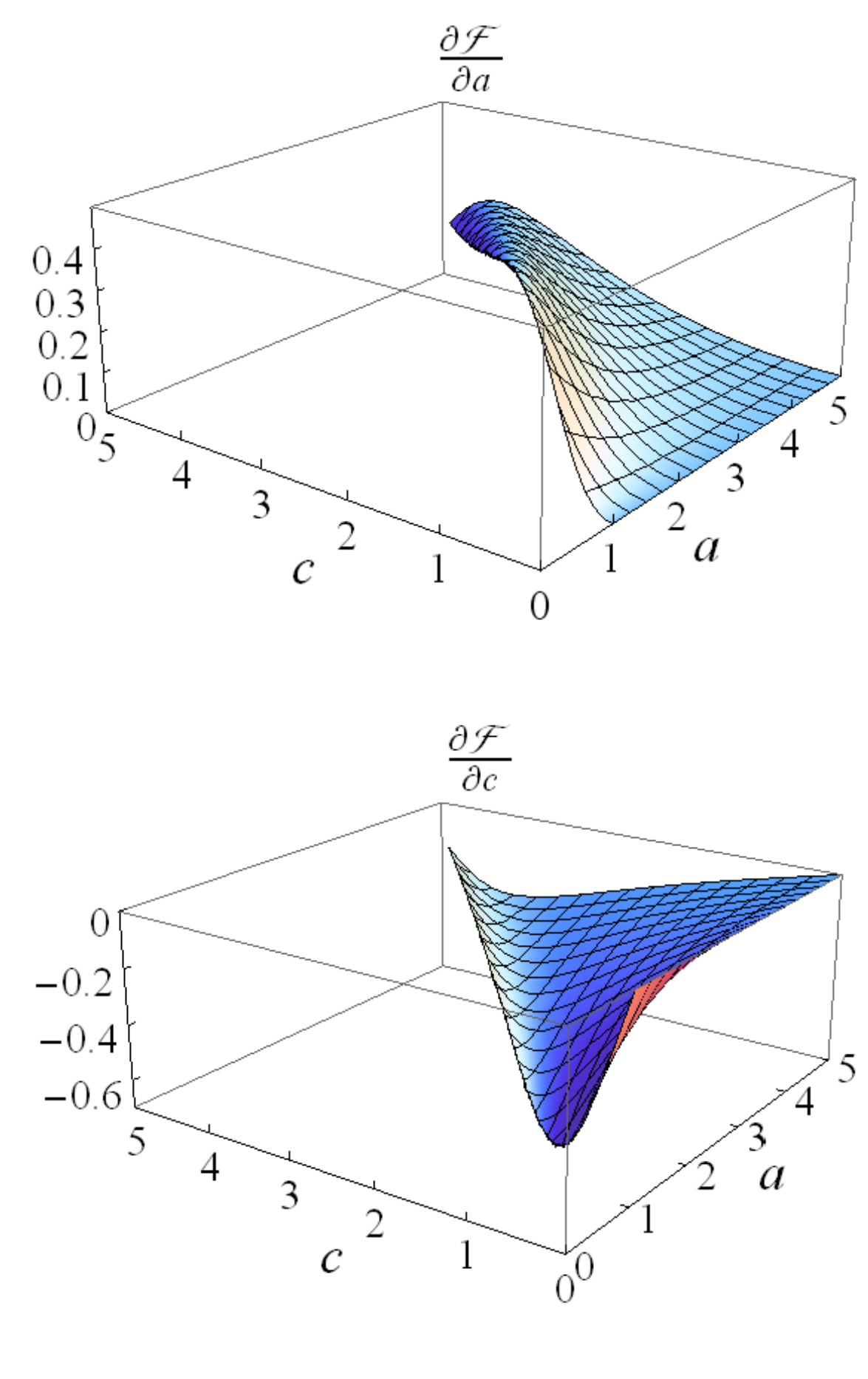}
\caption{Derivatives of the Uhlmann fidelity $\c F$ over the entries $a$ and $c$ of the covariance matrix, Eq.~(\ref{sqthcorm}), in the range of values corresponding to physical states, Eq.~(\ref{uncert}). The partial derivative $\frac{\partial \c F}{\partial a}\big|_{c}$ is always positive while   $\frac{\partial \c F}{\partial c}\big|_{a}$ is always negative.}
\label{dfdx}
\end{figure}

Also in the case of $\c F$ the derivative over the diagonal entry of the covariance matrix at constant off-diagonal elements is positive, while the derivative over the off-diagonal entries at constant diagonal entries is negative.


In the case of non-symmetric $STSs$ discussed in Sec.~\ref{noiseenhancedQCB}, the partial derivative of $QCB$, Eq.~(\ref{derivchernoffas}), with
respect to $N_{th_1}$ at constant $N_{th_2}$ and $r$ reads
\begin{equation}
\frac{d QCB}{d N_{th_1}}\big|_{N_{th_2},r}=-(N_{th_1}-N_{th_2})g \; ,
\label{derivovernth1}
\end{equation}
where
\begin{widetext}
\begin{equation}
g=\frac{8 (N_{th_1}+N_{th_2}+1) (2 N_{th_2}+1) \sinh ^2(2 r)}{\left(N_{th_1}^2-2 (7 N_{th_2}+3) N_{th_1}+(N_{th_2}-6)
   N_{th_2}-(N_{th_1}+N_{th_2}+1)^2 \cosh (4 r)-3\right)^2} \; .
\end{equation}
\end{widetext}
Since $g\geq 0$ the derivative, Eq.~(\ref{derivovernth1}), is positive only if $N_{th_1} < N_{th_2}$, negative only if $N_{th_1} > N_{th_2}$, and vanishes identically for symmetric STSs, namely for $N_{th_1} = N_{th_2}$, yielding a noise-independent $QCB$.

From the above results it follows that in the range  $N_{th_1} > N_{th_2}$ the $QCB$ is a monotonically decreasing function of the number of thermal photons in the first mode (increasing local thermal noise) and vanishes asymptotically, together with the probability of error, as $N_{th_1} \rightarrow \infty$.

\end{document}